\documentclass[12pt]{article}

\usepackage{arxiv_package}

\usepackage{mathrsfs}

\usepackage{enumitem}

\usepackage[utf8]{inputenc} 
\usepackage[T1]{fontenc}    
\usepackage{url}            
\usepackage{hyperref}       
\usepackage{booktabs}       
\usepackage{amsfonts}       
\usepackage{nicefrac}       
\usepackage{microtype}      
\usepackage{bm}
\usepackage{amsthm}
\usepackage{comment}
\usepackage{cleveref}
\usepackage{amsmath}
\usepackage{multicol, multirow, colortbl}
\usepackage{pifont}
\usepackage{caption}
\usepackage{subcaption}
\usepackage{tabularx}
\usepackage{multirow}
\usepackage{makecell}
\usepackage{wrapfig}
\usepackage{diagbox}
\usepackage{float}

\usepackage{appendix}

\usepackage{bbm}
\usepackage{dsfont}

\DeclareMathAlphabet{\mathbbold}{U}{bbold}{m}{n}
\newcommand*{\one}{\mathbbold{1}}

\newcommand{\quotes}[1]{``#1''}

\hypersetup{
    colorlinks=true,
    citecolor = blue,
    linkcolor=blue
}

\makeatletter
\let\original@footnotemark\footnotemark
\newcommand{\align@footnotemark}{%
  \ifmeasuring@
    \chardef\@tempfn=\value{footnote}%
    \original@footnotemark
    \setcounter{footnote}{\@tempfn}%
  \else
    \iffirstchoice@
      \original@footnotemark
    \fi
  \fi}
\pretocmd{\start@align}{\let\footnotemark\align@footnotemark}{}{}
\makeatother

\def\independenT#1#2{\mathrel{\rlap{$#1#2$}\mkern2mu{#1#2}}}

\usepackage{mathtools}

\newcolumntype{C}{>{\centering\arraybackslash}X}

\usepackage{mathtools}

\def\calS{\mathcal{S}}

\usepackage[algoruled]{algorithm2e}

\def\1{\mathbbm{1}}

\usepackage[normalem]{ulem} 
\usepackage{cancel}

\def\RRT{{\text{RT (RandomSplit) }}}
\def\ART{{\text{RT (AdaSplit) }}}

\def\weight{{\xi}}

\title{Adaptive sample splitting for randomization tests}

		\author[]{Yao Zhang\footnote{Department of Statistics, Stanford University} \ \ and  \  Zijun Gao\footnote{Department of Data Science and Operations, University of Southern California} }
	
\begin{document}

\maketitle

\begin{abstract}
Randomization tests are widely used to generate finite-sample valid 
$p$-values for causal inference on experimental data. However, when applied to subgroup analysis, these tests may lack power due to small subgroup sizes. Incorporating a shared estimator of the conditional average treatment effect (CATE) can substantially improve power across subgroups but requires sample splitting to preserve validity. To this end, we quantify each unit's contribution to estimation and testing using a certainty score, which measures how certain the unit's treatment assignment is given its covariates and outcome.
We show that units with higher certainty scores are more valuable for testing but less important for CATE estimation, since their treatment assignments can be accurately imputed. Building on this insight, we propose AdaSplit, a sample splitting procedure that adaptively allocates units between estimation and testing to maximize their overall contribution across tasks.
We evaluate AdaSplit through simulation studies, demonstrating that it yields more powerful randomization tests than baselines that omit CATE estimation or rely on random sample splitting. Finally, we apply AdaSplit to a blood pressure intervention trial, identifying patient subgroups with significant treatment effects.
\end{abstract}

\section{Introduction}

\subsection{Subgroup analysis in randomized controlled trials}

Subgroup analysis of heterogeneous treatment effects is crucial for evaluating treatment efficacy and safety across all phases of clinical trials \citep{wang2007statistics,rothwell2005subgroup}.
In early-phase trials, such analyses explore treatment effects across diverse patient types, helping to refine later-phase studies and patient selection criteria \citep{lipkovich2011subgroup,seibold2016model,friede2018recent}.
In confirmatory trials, they evaluate the effect consistency across patient subgroups, supporting regulatory review and benefit–risk assessment
\citep{tanniou2016subgroup,amatya2021subgroup,paratore2022subgroup}.
Overall, these analyses offer insights into patients who can be helped or harmed by treatment, guiding both trial design and real-world application.

\begin{figure}[t]
    \centering
        \includegraphics[width=0.8\textwidth]{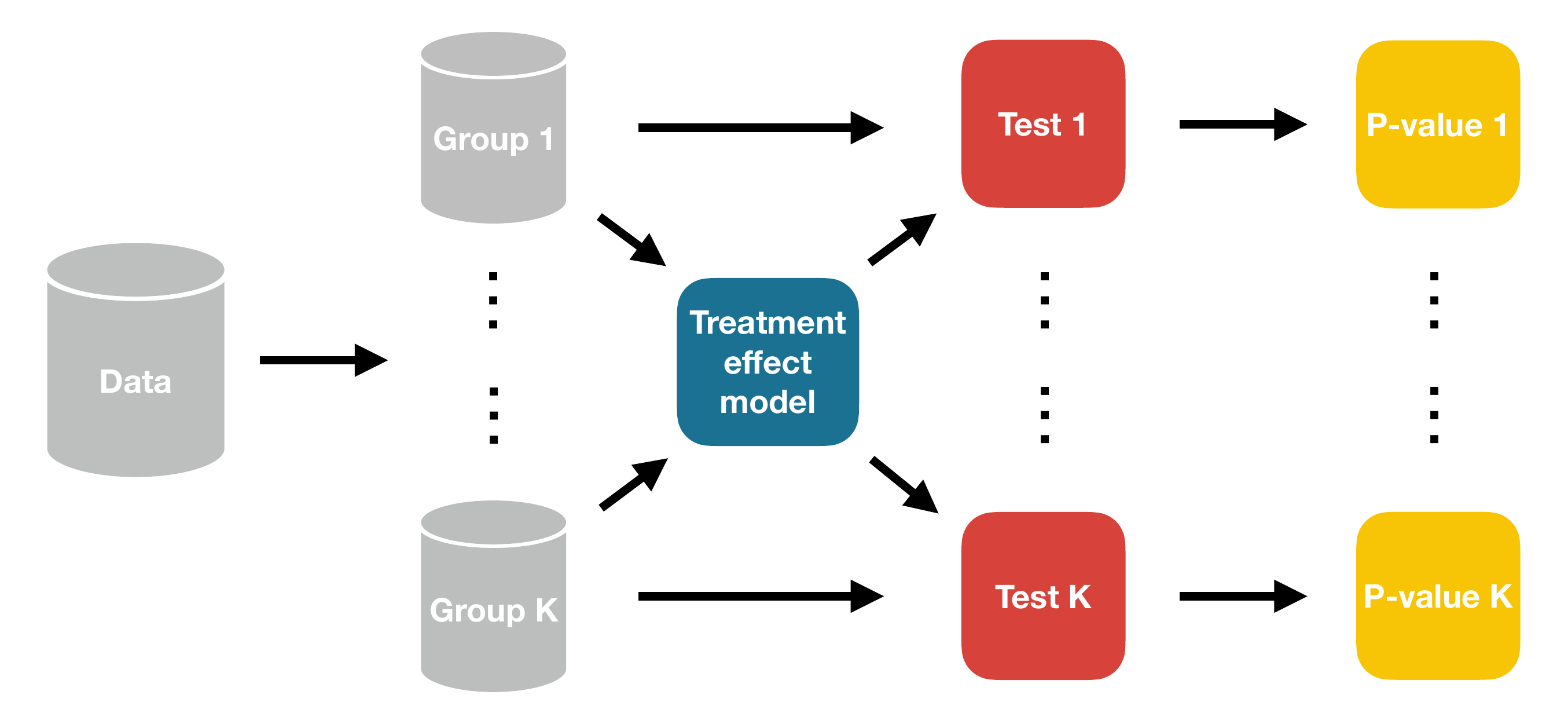}
    \caption{Diagram of adaptive sample splitting (AdaSplit) for randomization tests in subgroup analysis. AdaSplit adaptively splits out a subset of units from each subgroup to fit a shared regression model for estimating the conditional average treatment effect. This model and the carefully retained units are then used to construct randomization tests, yielding powerful $p$-values for testing treatment effects within each subgroup. }
    \label{fig:diagram}
\end{figure}

Although widely considered in clinical research, subgroup analysis methods face several common challenges that can lead to misleading results in practice.
First, methods that rely on strong modelling assumptions may falsely detect treatment effects when models are overfitted or misspecified \citep{athey2015machine,burke2015three}.
Second, methods for post-hoc subgroup selection may introduce bias, requiring further correction to ensure validity \citep{thomas2017comparing,guo2021inference}.
Third, methods for conducting subgroup analyses across multiple baseline covariates are often informal or overly stringent, failing to adequately address multiplicity across $p$-values \citep{lagakos2006challenge,wang2007statistics,bailar2012medical}.

In this article, we address the validity challenges in subgroup analysis from a fresh perspective using Fisher randomization tests \citep{rubin1980randomization,fisher1935design}. These tests offer two key advantages to statistical inference more broadly. First, they compute $p$-values based on the known assignment distribution, guaranteeing Type I error control without relying on any model assumptions. Second, they allow flexible test statistics to capture different types of treatment effects \citep{caughey2023randomisation,zhang2025multiple}, and can incorporate machine learning models to boost power, as long as the computational budget permits \citep{guo2025ml}. Building on these strengths, we develop a new randomization test framework for subgroup analysis.

\subsection{Our Contributions}\label{sect:outline}

\begin{figure}[t]
    \centering
    \hspace{35pt}
        \includegraphics[width=0.85\textwidth]{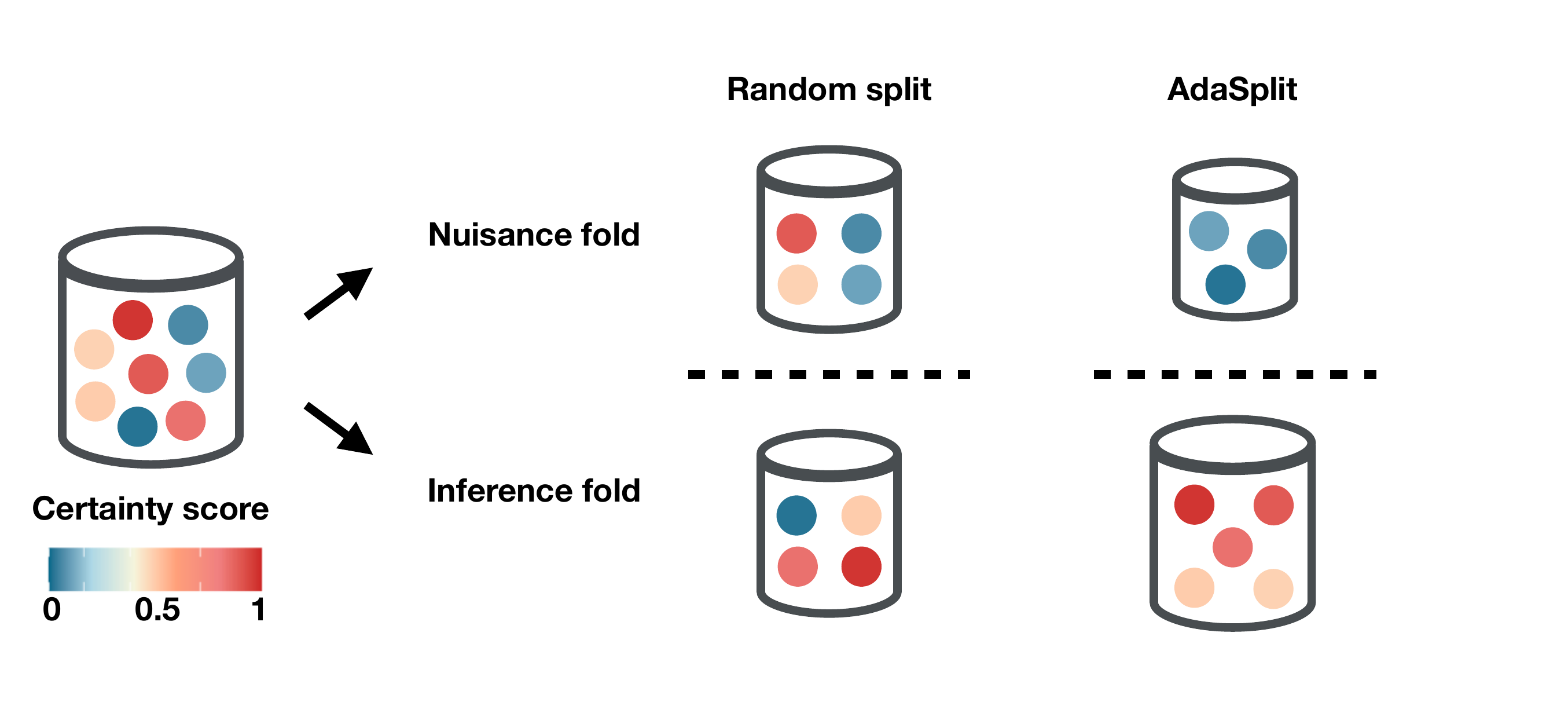}
    \caption{Comparison of random sample splitting and adaptive sample splitting (AdaSplit). AdaSplit allocates units into the nuisance fold (for CATE estimation) and the inference fold (for testing) based on their certainty scores defined in \eqref{equ:c_i}.  }
    \label{fig:diagram_2}
\end{figure}

\Cref{fig:diagram} illustrates our proposed framework for subgroup analysis. In a randomized controlled trial,  $n$ units are divided into $K$ pre-specified, disjoint subgroups based on their covariate information, and randomization tests are used to detect treatment effects within each subgroup. This multiple testing problem is formalized in \Cref{sec:formulation}.
As we will see, the power of these tests can be greatly enhanced by incorporating a shared estimator of the conditional average treatment effect (CATE) across subgroups. However, if the CATE estimators and test statistics are constructed using treatment assignments from the same units, the resulting p-values may no longer be valid.

To address this, we propose an adaptive sample splitting procedure, AdaSplit, which allocates the units into a nuisance fold $\mathcal I\subset  [n] := \{1,\dotsc, n\}$ for CATE estimation and an inference fold $\mathcal J = [n]\setminus \mathcal I$ for hypothesis testing. As shown in \Cref{fig:diagram_2}, AdaSplit's allocation strategy is primarily guided by a certainty score for each unit’s treatment assignment $Z_i$, derived from its covariates $X_i$ and outcome $Y_i$. 
Under a Bernoulli trial with probability $1/2$, this certainty score is defined as
\begin{equation}\label{equ:c_i}
 C_i := |2e(X_i,Y_i)-1|,
\end{equation}
where $e(X_i,Y_i) := \mathbb P\{Z_i = 1 \mid X_i,Y_i\}$ is the posterior assignment probability.
When $C_i = 0$, meaning $e(X_i, Y_i) = 1/2$, we are most {\it uncertain} about the value of $Z_i$. When $C_i = 1$, meaning $e(X_i, Y_i) = 0$ or 1,
we are most  {\it certain} about the value of $Z_i$. In between, $C_i\in (0,1)$ quantifies how confidently  $Z_i$ can be predicted from $X_i$ and $Y_i$.

In \Cref{sect:ass}, we characterize units' contribution to estimation and testing using $C_i$:
\begin{enumerate}[label=(\alph*)]
\item Units with {\it larger} $C_i$ contribute more to the asymptotic test power (\Cref{sec:influential.RT}).
\item Units with {\it smaller} $C_i$ are more important for CATE estimation, as the assignments of other units can be imputed in our proposed estimator (\Cref{sect:bar}).
\end{enumerate}
We leverage this observation to develop a practical algorithm of  AdaSplit in \Cref{sect:alg}. The algorithm gradually expands the nuisance fold $\mathcal I$ by adding units in increasing order of estimated $C_i$, until the CATE estimates trained on $\mathcal I$ meet a convergence criterion. The algorithm offers three desirable properties:
\begin{enumerate}
     \item It refrains from using any assignment from the inference fold, yielding valid and independent $p$-values for (multiple) hypothesis testing (\Cref{sect:valid}).
    \item It aims to maximize test power by reserving units with large $C_i$ for the inference fold, in line with point (a) highlighted above.
    \item It is fully deterministic: unlike random sample splitting, it produces the same nuisance and inference folds in every run.
\end{enumerate}

In \Cref{sect:exp}, we conduct experiments to evaluate the performance of AdaSplit across various settings. Our results show that AdaSplit produces more powerful $p$-values than standard baselines that either ignore CATE estimation or rely on random sample splitting.
We further demonstrate AdaSplit on the Systolic Blood Pressure Intervention Trial (SPRINT) dataset \citep{Wright2016ART}, identifying interpretable patient subgroups with strong treatment effects.
Finally, in \Cref{sect:discussion}, we discuss how our adaptive sample splitting idea can be extended in future work.

\subsection{Related work}\label{sect:related}

We next discuss related work in three different areas: subgroup analysis in clinical trials, randomization tests for causal inference, and active learning for classification.

The literature on subgroup analysis is often divided into exploratory methods, which aim to identify subgroups with differential treatment effects \citep{su2009subgroup,Shen2015InferenceFS,seibold2016model,li2023statistical}, and confirmatory methods, which seek to validate subgroups identified in earlier stages \citep{jenkins2011adaptive,friede2012conditional,guo2021inference}. Exploratory approaches typically focus on developing data-driven techniques (such as mixture models or tree-based algorithms) to partition the covariate space into subgroups with different levels of treatment effects. 
In contrast, our method focuses on settings with pre-specified subgroups, as commonly required in confirmatory trials by regulatory agencies, which mandate analyses across key demographic strata such as age, race, or sex.
Compared to existing confirmatory approaches, our method is distinguished by its use of randomization tests, which provide finite-sample valid $p$-values without relying on modelling assumptions. Furthermore, our procedure does not require bias correction for validity, as sample splitting naturally prevents data dredging. To reduce the efficiency loss from sample splitting, we solve a special technical problem: how to adaptively allocate units between estimation and testing to preserve the power of randomization tests. This is different from those addressed in prior work, such as subgroup discovery or post-selection bias correction.

Our article is also related to a growing body of work that extends randomization tests to hypotheses weaker than the original sharp null of no treatment effect on any unit. For example, \citet{fogarty2021prepivoted,cohen2020gaussian,zhao2021covariate} propose valid tests for the weak null of zero average treatment effect.
Other works, such as \citet{caughey2023randomisation} and \citet{chen2024enhanced}, show that test statistics satisfying certain properties can produce valid inference for quantiles of individual treatment effects. 
While these papers focus on constructing a single randomization test, \citet{zhang2025multiple} propose a framework for conducting multiple randomization tests for lagged and spillover effects in complex experiments.
However, none of these works focus on testing subgroup effects, which can be viewed as relaxing the sharp null hypothesis from all units to specific subgroups. Our work fills this gap by introducing an adaptive method. In this sense, our contribution is orthogonal to existing methods and may enhance their power when applied to subgroup analyses.

Finally, the high-level idea of our adaptive sample splitting procedure is related to active learning, a subfield of machine learning that studies how to efficiently build accurate classifiers by selecting the most informative data points for labelling. Many active learning algorithms prioritize querying points with high predictive uncertainty or disagreement among classifiers \citep{schohn2000less,balcan2006agnostic,hanneke2014theory,ash2019deep}. Related lines of work include coreset selection, which aims to identify a representative subset of data that can train a model with accuracy comparable to the one trained on the entire dataset \citep{wei2015submodularity,sener2017active,rudi2018fast,borsos2024data}.
Unlike these approaches, our splitting criterion is not purely driven by model quality or computational efficiency. Instead, it is tailored to maximize the power of randomization tests. Furthermore, in contrast to coreset methods that can examine the entire dataset during selection, our algorithm remains blind to the treatment assignments in the inference fold throughout iterations, which preserves the validity of our randomization tests.

\subsection{Notation \& Assumptions}\label{sect:notation}

Each unit \( i \in [n] \) is associated with a set  of covariates \( X_i \in \mathbb{R}^d \), a treatment assignment variable \( Z_i \in \{0,1\} \), and two potential outcomes \( Y_i(0), Y_i(1) \in \mathbb{R} \). The following are the three basic assumptions considered in this article.
\begin{assumption}\label{ass:randomization}
Each treatment assignment \( Z_i \) is drawn from a Bernoulli distribution \( \text{Bern}(e(X_i) := \mathbb{P}(Z_i = 1 \mid X_i)) \), where $e(X_i) := \mathbb{P}\{Z_i = 1 \mid X_i\}$. Unless otherwise specified, we consider the Bernoulli design with $e(X_i) =  1/2$ for all $i\in [n]$, that is,
\[
\text{Bern(1/2) design}: \  Z_1,\dotsc, Z_n \overset{\mathrm{i.i.d.}}{\sim} \text{Bern}(1/2).
\]
\end{assumption}
\vspace{5pt}
\begin{assumption}\label{ass:consistency}
The observed outcome \( Y_i = Z_i Y_i(1) + (1 - Z_i) Y_i(0) \) for all \( i \in [n] \).
\end{assumption}
The second assumption originates from the causal model of \citet{rubin1974estimating}, which implies that the treatment has no hidden variation or interference across units. Based on this assumption, 
we define the conditional average treatment effect (CATE) as
\[
\tau(x) = \mathbb E[Y_i(1) - Y_i(0)\mid X_i=x] = \mu_1(x) -  \mu_0(x),
\]
where  $\mu_z(x) = \mathbb E[Y\mid X =x, Z=z]$\footnote{Besides \Cref{ass:randomization}, CATE estimation in observational studies requires additional assumptions (unconfoundedness and positivity) to achieve the second equality. In our setup, it is assumed that these assumptions are satisfied by the randomization of treatment assignments.}
is the expected outcome when $X=x$ and $Z = z$.  

Using $\tau(x)$ and $\mu(x) := \mathbb E[Y\mid X=x]$, we can express $\mu_z(x)$ as
\begin{equation}\label{equ:two_mu}
 \mu_{z}(x) =  \mu(x) + [z-e(x)]  \tau(x),~ z\in\{0,1\}.
\end{equation}

\begin{assumption}\label{example:1}
    The outcomes $Y_i,i\in [n],$ follow a Gaussian model:
    \begin{equation}\label{equ:normal_model}
Y_i = \mu_0(X_i) + Z_i \tau(X_i) + \epsilon_i,\quad \epsilon_i \sim \mathcal N (0, \nu^2 ),\quad \epsilon_i \independent X_i, Z_i,~\forall i\in [n].
\end{equation}
\end{assumption}

Finally, for any subset $\mathcal{S} \subseteq [n]$, we use the subscript notation to denote restriction to the indices in $\mathcal{S}$. For example, $X_{\mathcal{S}} := (X_i)_{i \in \mathcal{S}}$.
We define $O_i := (X_i, Y_i, Z_i)$ and $\tilde{O}_i := (X_i, Y_i, \tilde{Z}_i)$, where $\tilde{Z}_i$ denotes a randomized treatment assignment drawn from the same distribution as $Z_i$. Similarly, let $O_{\mathcal{S}} := (O_i)_{i \in \mathcal{S}}$ and $\tilde{O}_{\mathcal{S}} := (\tilde{O}_i)_{i \in \mathcal{S}}$.
In the randomization tests introduced below, we let $\tilde{\mathbb{P}}$ denote the distribution of the randomized treatment assignments $\tilde{Z}_i$. We refer to the distribution of the test statistic $T(\tilde{O}_{[n]})$ under $\tilde{O}_{[n]} \sim \tilde{\mathbb{P}}$ as the reference distribution.

\section{Randomization tests for subgroup analysis}\label{sec:formulation}

In this section, we briefly introduce randomization tests, describe their extension to subgroup analysis, and highlight the key challenges involved.

\subsection{Randomization tests}

The individual treatment effect $Y_i (1) - Y_i(0)$ is unobserved in the data, as unit $i$ is either treated (with the control outcome missing) or assigned to control (with the treated outcome missing). To address this fundamental challenge of causal inference, randomization tests are commonly used to test \citet{fisher1935design}’s sharp null hypothesis,
\begin{equation}\label{equ:h_0_global}
    H_0: Y_i(1)=Y_i(0),~\forall i \in [n],
\end{equation}
This hypothesis states that the treatment has zero effect on every unit $i$.  It can be easily extended to test whether the treatment has a constant effect across all units. For simplicity, we focus on zero-effect hypotheses like $H_0$ throughout this article.

Although $H_0$ is a strong assumption (as we will discuss later), it allows us to impute the missing potential outcomes for all units; that is, $Y_i(1) = Y_i(0) = Y_i(Z_i) = Y_i$, where $Y_i$ is the observed outcome for unit $i$.
To define a randomization test for $H_0$, we first introduce a test statistic $T(O_{[n]})$, which maps the observed data to a summary measure---typically a treatment effect estimate---that provides evidence against the zero-effect hypothesis $H_0$. Technically, there is no restriction on the choice of test statistic $T$, as long as it is computable from the observed data.

The randomization test then computes a $p$-value 
$P(O_{[n]})
    := \mathbb {\tilde P}\{ T(\tilde O_{[n]}) \geq  T(O_{[n]}) \}$ by comparing the observed statistic $T(O_{[n]})$ with the reference distribution of $T(\tilde{O}_{[n]})$.
For example, in the Bern(1/2) design, this $p$-value is given by 
\[
P(O_{[n]})
    = 2^{-n} \sum_{{\tilde z}_{[n]} \in \{0,1\}^n} \1 
    \{T(X_{[n]}, Y_{[n]},\tilde z_{[n]}) \geq T(O_{[n]})\}.
\]
The $p$-value $P(O_{[n]})$ represents the proportion of randomized treatment assignments that yield a test statistic at least as large as the observed one. A small $p$-value suggests that the treatment effect reflected in the observed statistics $T(O_{[n]})$ is unlikely to have occurred by chance. This construction guarantees Type I error control:
\begin{equation}\label{equ:validity}
\mathbb P_{H_0}\left\{P(O_{[n]})\leq \alpha \mid X_{[n]},Y_{[n]}(0),Y_{[n]}(1)\right\} \leq \alpha,\ \forall \alpha\in [0,1],
\end{equation}
without making any assumptions about $X_{[n]}$, $Y_{[n]}(0)$ and $Y_{[n]}(1)$.

\vspace{-5pt}
\subsection{Subgroup-based randomization tests}\label{sect:vanilla}

From a critical perspective, rejecting the sharp null $H_0$ merely indicates that the treatment has a nonzero effect for at least one unit. It provides no information about which units are more likely to benefit from the treatment---an insight that is often essential in real-world applications concerned with heterogeneous treatment effects.

One natural way to relax Fisher's sharp null hypothesis is to divide the $n$ units into $K$ disjoint subgroups and test a sharp null hypothesis within each subgroup. For example, given a partition of the covariate space $\mathbb{R}^d = \bigcup_{k \in [K]} \mathcal{X}_k$, we define subgroup $k$ as
$\mathcal{S}_k := \{i \in [n] : X_i \in \mathcal{X}_k\}$,
and test the null hypothesis
\begin{equation}\label{equ:null}
    H_{0,k} : Y_i(1) = Y_i(0),~\forall i \in \mathcal{S}_k.
\end{equation}
Testing these subgroup-specific hypotheses can uncover the effect heterogeneity across covariates, e.g., age or biomarkers, and inform more personalized treatment plans.

\textbf{Power loss.}
However, randomization tests for the subgroup nulls may lack power, especially when using simple test statistics such as difference-in-means (DM):
\[
T_{\text{DM}}(\tilde{O}_{\mathcal{S}_k}) := \frac{2}{|\mathcal{S}_k|} \bigg\{ \sum_{i \in \mathcal{S}_k} \tilde{Z}_i Y_i - \sum_{i \in \mathcal{S}_k} (1 - \tilde{Z}_i) Y_i \bigg\}.
\]
Under the $\text{Bern}(1/2)$ design, the statistics $T_{\text{DM}}(\tilde{O}_{\mathcal{S}_k})$  has mean zero and variance
\[
\Var [T_{\text{DM}}(\tilde{O}_{\mathcal{S}_k})\mid Y_{\mathcal{S}_k} ] = 4|S_k|^{-2}\sum_{i\in \mathcal S_k}Y_i^2 = O_{\mathbb{P}}(|\mathcal{S}_k|^{-1}).
\]
The variance of $T_{\text{DM}}(O_{\mathcal{S}_k})$ is also of the same order. The mean difference between the two statistics is $O_{\mathbb{P}}(1)$. As the subgroup size $|\mathcal{S}_k|$ shrinks, their distributions overlap more, inflating the $p$-value  $\tilde{\mathbb P} \{ T_{\text{DM}}(\tilde O_{\mathcal{S}_k}) \geq T_{\text{DM}}( O_{\mathcal{S}_k}) \} $.

To reduce variance and improve power in randomization tests, one can incorporate regression models for covariate adjustment \citep{tsiatis2008covariate,lin2013agnostic,rothe2018flexible,guo2023generalized}. For example, \citet{rosenbaum2002covariance} defines a difference-in-means statistic using residuals from an outcome model $\hat \mu$, while \citet{zhao2021covariate} recommend using a linear model with treatment-covariate interactions to construct robust $t$-statistics. These approaches are proposed to test the sharp null $H_0$ in \eqref{equ:h_0_global} and the weak null of zero average treatment effect.

In contrast, our goal is to construct powerful randomization tests for subgroup-specific null hypotheses, which provide more granular insight into treatment effect heterogeneity. While the inferential targets differ, our approach is technically similar: given the functions $\mu$ and $\tau$, we construct a model-assisted test statistic
\[
T_{\text{AIPW}}(\tilde O_{\calS_k}) = |\mathcal{S}_k|^{-1} \sum_{i\in \mathcal{S}_k} \phi_{\text{AIPW}} (\tilde O_i;e, \mu, \tau),
\]
using the well-known augmented inverse probability weighting (AIPW) formula \citep{robins1994estimation} for the average treatment effect (ATE):
\begin{equation} \label{eq:AIPW}    
\begin{split}
\phi_{\text{AIPW}} (\tilde O_i;e, \mu, \tau)  =  
\frac{\tilde Z_i}{e(X_i)} \left[ Y_{i} -  \mu_1(X_{i}) \right]  -   \frac{1-\tilde Z_i}{1-e(X_i)} \left[ Y_{i} -  \mu_0(X_{i}) \right]      +  \tau (X_i),
\end{split}
\end{equation}
where $\mu_0$ and $\mu_1$ are defined using $\mu$ and $\tau$ as in \eqref{equ:two_mu}. Under the $\text{Bern}(1/2)$ design, the AIPW statistic $T_{\text{AIPW}}(\tilde O_{\mathcal{S}_k})$ has mean zero and variance
\[
\Var [T_{\text{AIPW}}(\tilde{O}_{\mathcal{S}_k})\mid X_{\mathcal{S}_k},Y_{\mathcal{S}_k} ] = 4|S_k|^{-2}\sum_{i\in \mathcal S_k}[Y_i-\mu(X_i)]^2,
\]
which is significantly smaller than the variance of $T_{\text{DM}}(\tilde{O}_{\mathcal{S}_k})$ if $\mu(X_i)$ explains away most of the variation in $Y_i$. Consequently, even when $|\mathcal{S}_k|$ is small, the distributions of $T_{\text{AIPW}}(O_{\mathcal{S}_k})$ and $T_{\text{AIPW}}(\tilde O_{\mathcal{S}_k})$ can remain separated, yielding a small $p$-value.

In practice, the nuisance functions $\tau$, $\mu_0$, and $\mu_1$ are unknown and must be estimated from data.
Using the same dataset to estimate them and to conduct randomization tests may lead to invalid $p$-values.
Moreover, the $p$-values across subgroups can become dependent, since the nuisance estimators rely on overlapping treatment assignments. This dependency violates the assumptions required by standard multiple testing procedures such as the Benjamini–Hochberg (BH) method \citep{benjamini1995controlling}. Cross-fitting\footnote{Cross-fitting randomly splits the data into multiple folds, using each fold for testing while using the remaining folds for estimation, and then aggregates the resulting $p$-values for inference.} \citep{schick1986asymptotically,bickel1988estimating,chernozhukov2018double} 
can produce valid $p$-values for each subgroup null hypothesis.
However, standard valid methods for combining these dependent $p$-values, such as taking twice their average, may not improve power \citep{ruschendorf1982random,meng1994posterior,vovk2020combining}.
By adaptively splitting the sample between estimation and testing, AdaSplit addresses these challenges and improves power over random splitting.

\section{Adaptive sample splitting (AdaSplit)}\label{sect:ass}

The splitting strategy in AdaSplit is motivated by points (a) and (b) in \Cref{sect:outline}. We begin by illustrating these observations using a synthetic example, shown in \Cref{fig:AdaSplit_example}; simulation details are provided in \Cref{sect:appendix_exp_figure}.

\begin{figure}[t]
    \centering
    \begin{subfigure}{0.44\textwidth}
        \centering
        \includegraphics[width=\linewidth]{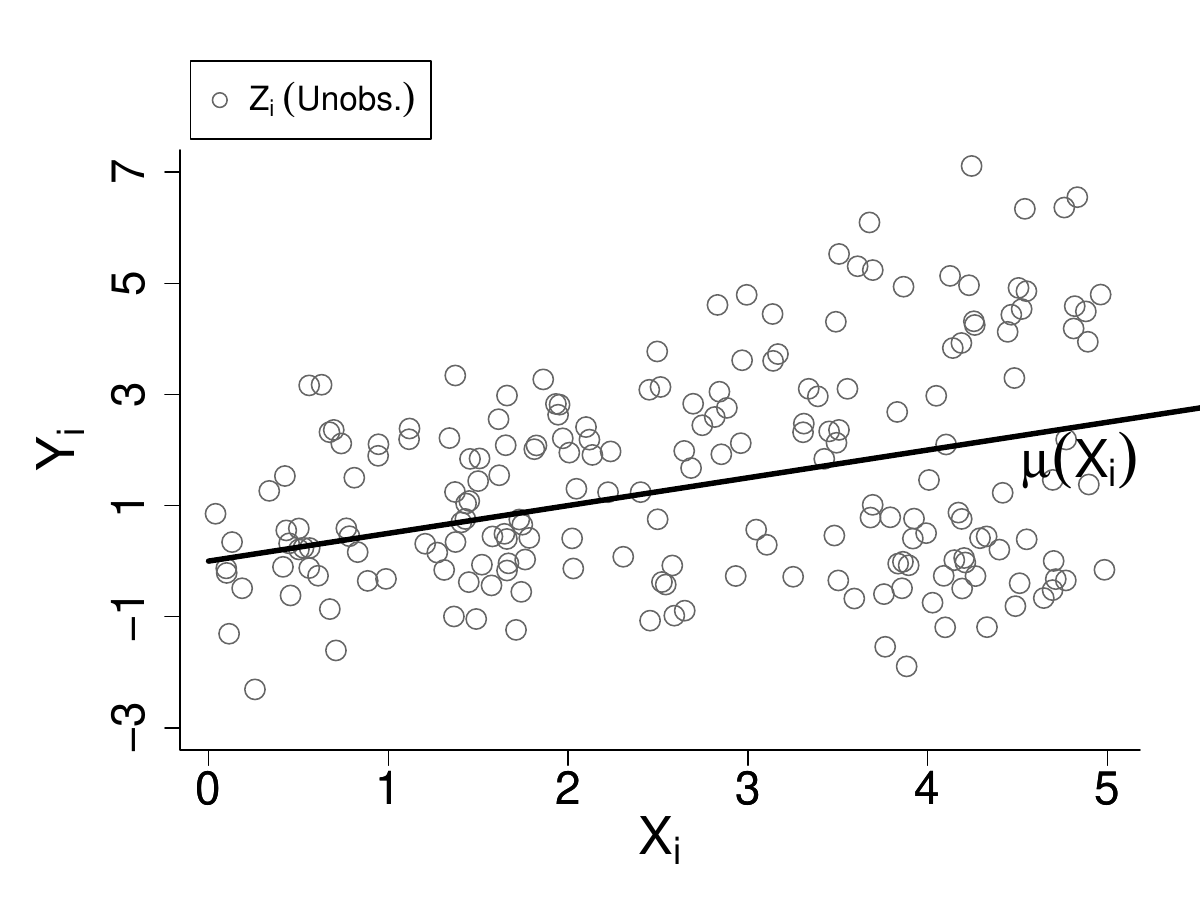}
        \caption{}\label{fig:AdaSplit_example_1}
    \end{subfigure}
    \begin{subfigure}{0.44\textwidth}
        \centering
        \includegraphics[width=\linewidth]{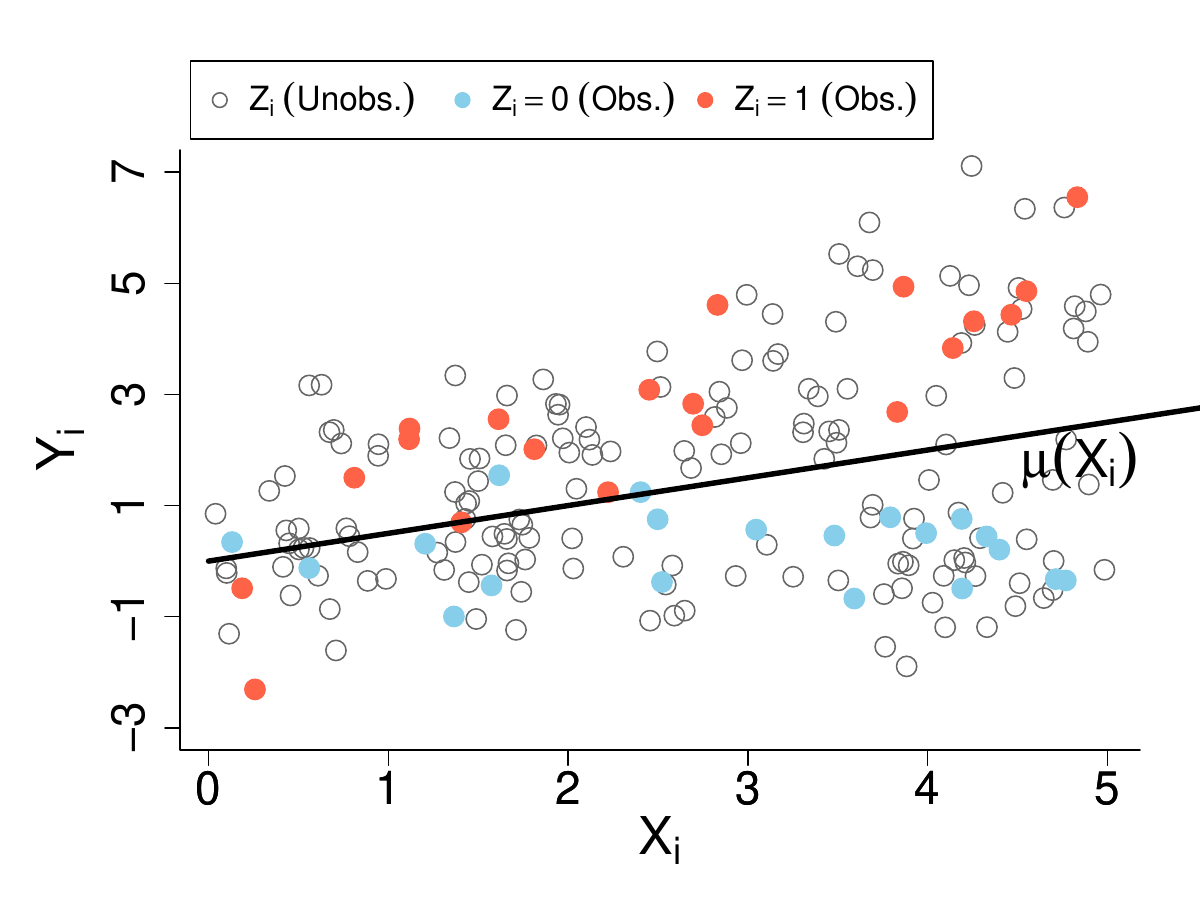}
        \caption{}\label{fig:AdaSplit_example_2}
    \end{subfigure}
    \begin{subfigure}{0.44\textwidth}
        \centering
        \includegraphics[width=\linewidth]{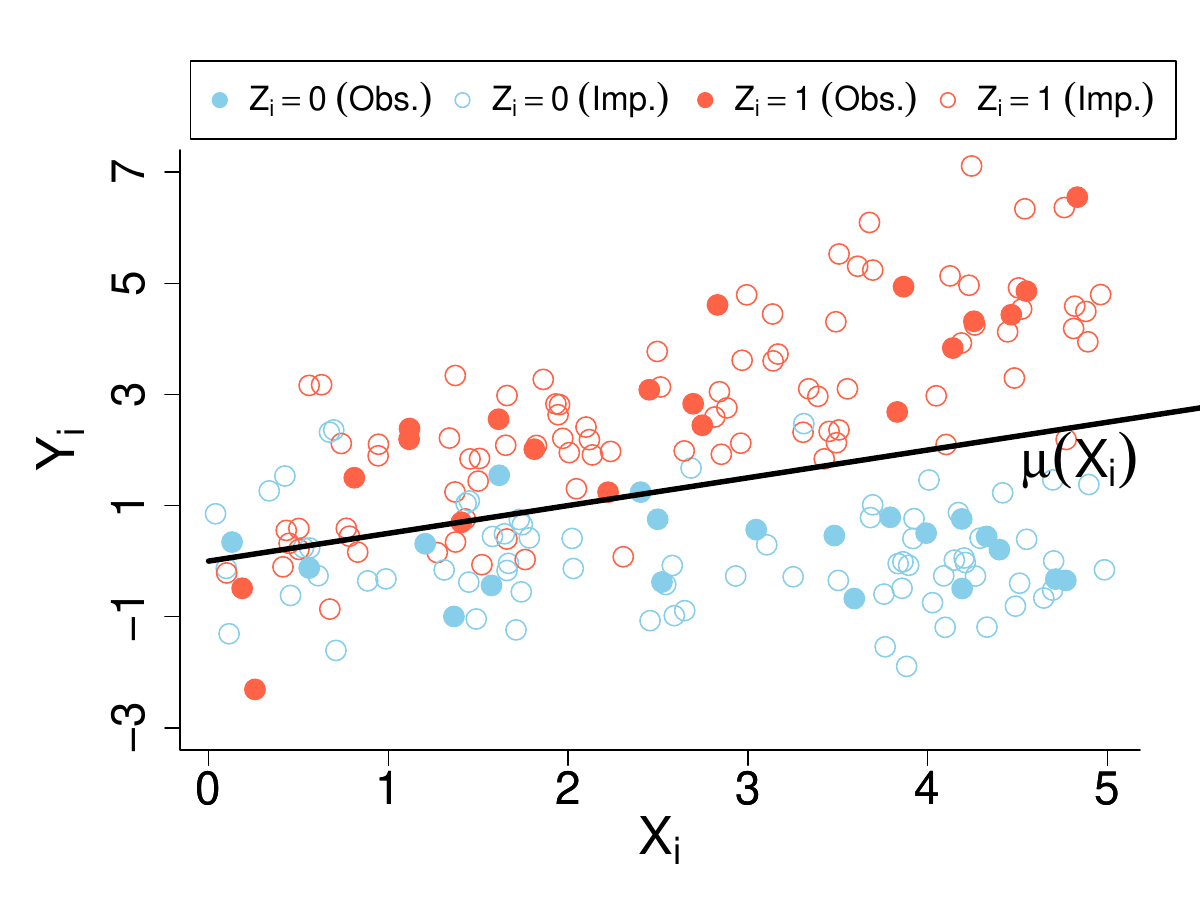}
        \caption{}\label{fig:AdaSplit_example_3}
    \end{subfigure}
    \begin{subfigure}{0.44\textwidth}
        \centering
        \includegraphics[width=\linewidth]{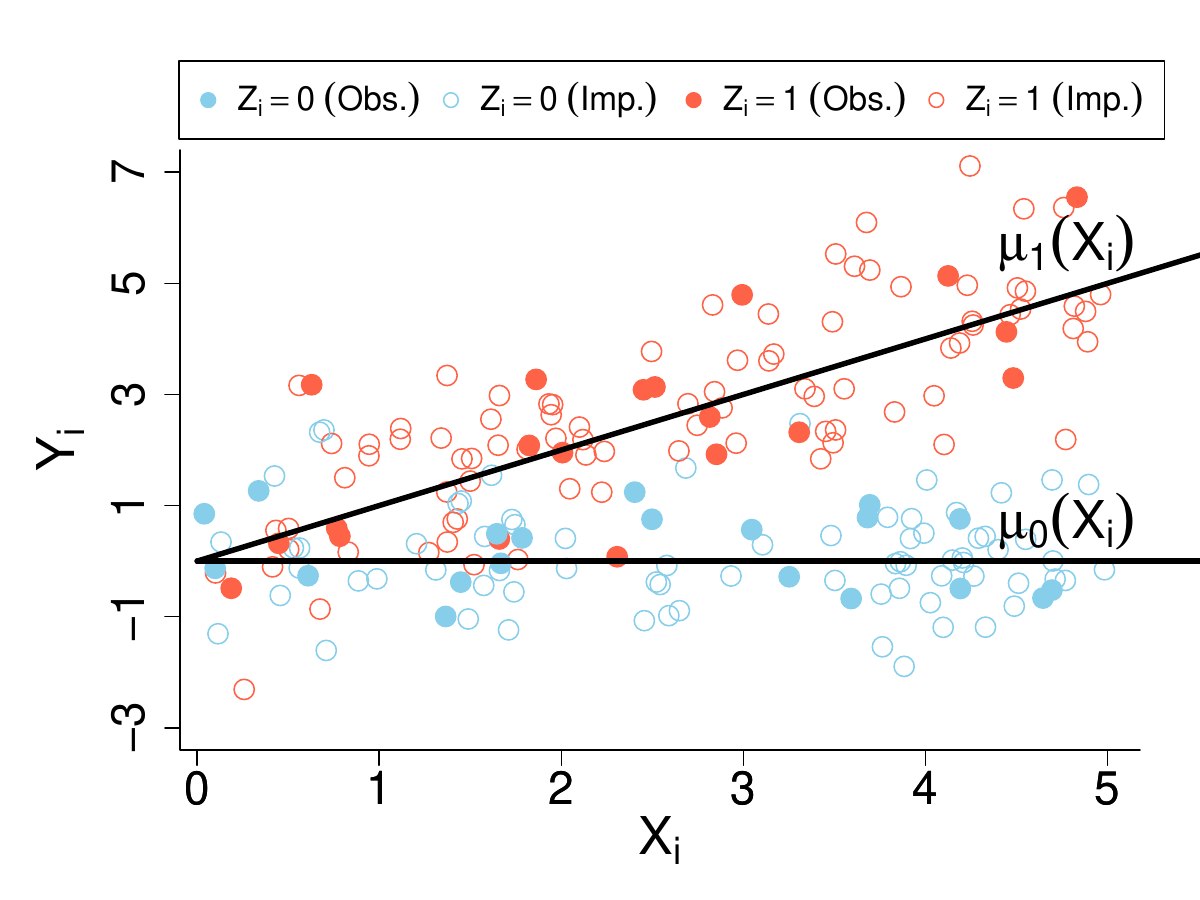}
        \caption{}\label{fig:AdaSplit_example_4}
    \end{subfigure}
\caption{
    Illustration of AdaSplit on a synthetic dataset. 
    (a) Estimate \( \mu(X_i) = \mathbb{E}[Y_i \mid X_i] \) without using treatment assignments. 
    (b) Reveal the treatment assignments for a subset of units  (solid points).
    (c) Impute the assignments of the remaining units (hollow points) using a model fitted to the revealed ones.
    (d) Use both observed and imputed assignments to estimate $\mu_z(X_i) = \mathbb E[Y_i \mid X_i,Z_i=z]$ in the AIPW statistics.
}
 \label{fig:AdaSplit_example}
\end{figure}

In \Cref{fig:AdaSplit_example_1}, each grey point is a data pair $(X_i, Y_i)$.  Without using any treatment assignment, we can construct an estimator $\hat \mu(X_i)$ of the outcome function $\mu(X_i) = \mathbb{E}[Y_i \mid X_i]$ by regressing $Y_{[n]}$ onto $X_{[n]}$. For simplicity, we assume $\hat \mu = \mu$ in this example. 

In \Cref{fig:AdaSplit_example_2}, we reveal the treatment assignments for a subset of units (solid points), with treated units shown in red and controls in blue. The treated units tend to have larger outcomes than the controls, suggesting that the CATE $\tau$ is a positive function. Substituting the expression for $\mu_0$ from \eqref{equ:two_mu} into \eqref{equ:normal_model} yields the residuals:
\[
Y_i - \mu(X_i) = [Z_i-1/2]\tau(X_i) +\epsilon_i.
\]
This implies that for units with unrevealed $Z_i$, those with 
$Y_i > \mu(X_i)$ are likely treated ($Z_i = 1$), while those with $Y_i < \mu(X_i)$ are likely controls ($Z_i = 0$). We can thus predict $Z_i$ from $(X_i, Y_i)$ using the posterior assignment probability $e(X_i, Y_i)$, derived in \Cref{prop:threshold}, which depends on the residual $Y_i - \mu(X_i)$. The further $Y_i$ deviates from its conditional expectation $\mu(X_i)$, the more accurately we can impute $Z_i$.

\Cref{fig:AdaSplit_example_3} shows the imputed values of the unrevealed $Z_i$ as the colors of the hollow points. In \Cref{fig:AdaSplit_example_4}, both observed and imputed $Z_i$ values are used to estimate the functions $\mu_0$ and $\mu_1$ for computing the AIPW statistic in \eqref{eq:AIPW}. As noted in point (b) of \Cref{sect:outline}, units with low certainty scores $C_i = |2e(X_i, Y_i) - 1|$ are harder to impute and should be assigned to the nuisance fold.

For point (a), we divide the units with large $C_i$ into two types: those with $e(X_i, Y_i) \approx 1$, which tend to have large outcomes $Y_i$ and appear in the treated group ($Z_i = 1$), and those with $e(X_i, Y_i) \approx 0$, which tend to have small outcomes $Y_i$ and appear in the control group ($Z_i = 0$). According to the AIPW formula in \eqref{eq:AIPW}, both types of units, when included in the inference fold, tend to increase the observed statistic relative to the randomized ones (with $\tilde{Z}_i \neq Z_i$), thereby yielding a powerful randomization test.

We now present the theories of AdaSplit, which formalize points (a) and (b).

\subsection{Conditional power analysis of randomization tests}\label{sec:influential.RT}

Let $\mathcal{J}_k  = \mathcal J \cap \mathcal S_k$ denote the inference fold in subgroup $\mathcal S_k$.
Given arbitrary estimators $\hat \mu$ and $\hat \tau_{\mathcal{I}}$ fitted to the data $(X_{[n]}, Y_{[n]}, Z_{\mathcal{I}})$, we construct
the AIPW statistic
\[
 T(\tilde O_{\mathcal{J}_k})  = \sum_{j\in \mathcal{J}_k} \phi_{\text{AIPW}} (\tilde O_j;e, \hat \mu,\hat \tau_{\mathcal{I}}),
\]
as in \eqref{eq:AIPW}. We then test the subgroup null hypotheses $H_{0,k}$ in \eqref{equ:null} using the p-value
\begin{equation}\label{equ:pk}
\hat P_k = \hat P_k(O_{\mathcal{J}_k})
    := \tilde{\mathbb P}\left\{ T(\tilde O_{\mathcal{J}_k}) \geq T( O_{\mathcal{J}_k}) \right\}.\footnote{The notation $\hat P_k$ indicates the $p$-value depends on the inference fold through $\hat \tau_{\mathcal I}$.} 
\end{equation}
For this $p$-value to be valid as in \eqref{equ:validity} and remain independent of those from other subgroups, $\mathcal{J}_k$ must be selected without using the assignments $Z_{\mathcal{J}}$ reserved for inference.
For example, choosing $\mathcal{J}_k$ to minimize \eqref{equ:pk} would invalidate the $p$-value, since $O_{\mathcal{J}_k}$ and $\tilde{O}_{\mathcal{J}_k}$ would no longer be drawn from the same distribution.

To address this, we \emph{marginalize} over the assignments $Z_{\mathcal{J}_k}$ in the $p$-value in \eqref{equ:pk},  treating them as unobserved prior to inference.
This yields a conditional $p$-value of the form
\begin{equation}\label{equ:p_k}
\hat P_k (X_{\mathcal{J}_k},Y_{\mathcal{J}_k})  = \mathbb E_{Z_{\mathcal J_k}} \left\{   \hat P_k( O_{\mathcal{J}_k})        \mid X_{\mathcal J_k},Y_{\mathcal J_k} \right\}.
\end{equation}
However, the conditional expectation in this $p$-value is intractable. We thus use its Gaussian approximation to guide the selection of the inference fold $\mathcal{J}_k$.

To describe the assumptions for this approximation, we observe that the gap between the observed and randomized statistics in \eqref{equ:pk} is driven by the difference in $\phi_{\text{AIPW}}$ evaluated at $O_j$ and $\tilde O_j$, which can be written as $\hat W_j(\tilde Z_j - Z_j)$, where
\begin{equation}\label{equ:w_j_hat}
\hat W_j = \frac{Y_j - \hat \mu_{\mathcal{I},1}(X_j)}{e(X_j)} + \frac{Y_j - \hat \mu_{\mathcal{I},0}(X_j)}{1 - e(X_j)},
\end{equation}
and the estimators $\hat \mu_{\mathcal{I},z}$ are derived from $\hat \mu$ and $\hat \tau_{\mathcal{I}}$, as defined in \eqref{equ:two_mu}.

\begin{assumption}\label{assumption:weight} It holds that
$ \sum_{j\in \mathcal{J}_k} \hat W_j^{3} /\big[\sum_{j\in \mathcal{J}_k} \hat W_j^{2}\big]^{3/2} = O_{\mathbb P }\left(|\mathcal J_k|^{-1/2}\right).$
\end{assumption}

\begin{assumption}\label{assumption:delta}
    There exists $\delta \in (0,1/2)$ such that $e(X_j,Y_j) \in [\delta,1-\delta],\forall j\in \mathcal{J}_k$.
\end{assumption}

\Cref{assumption:weight} ensures that the weights $\hat W_j$ are not too heavy-tailed, e.g., when they are of comparable magnitude. \Cref{assumption:delta} allows us to merge the variance and third-moment terms in the Berry–Esseen bound. Under these assumptions, the $p$-value in \eqref{equ:p_k} can be approximated by the bivariate Gaussian integral
$\mathbb{E}_{T_k}\big[ \tilde{\mathbb P}_{\tilde T_k} \{\tilde T_k \geq T_k \mid T_k \} \big] $
where $T_k$ and $\tilde T_k$ are the Gaussian limits of the observed and randomized statistics.

To simplify the analysis below, we consider a special case of our general result in \Cref{sect:proof_thm_1}. assuming $\mu$ and $\tau$ are known. This removes the dependence of the optimal choice of $\mathcal{J}_k$ on estimation error from the nuisance fold $\mathcal{I}$.

\begin{theorem}\label{thm:expected_power}
Suppose \Cref{ass:randomization,ass:consistency,example:1,assumption:weight,assumption:delta} hold with $\hat \mu = \mu$ and $\hat \tau_{\mathcal I } = \tau$. Then the conditional $p$-value $\hat P_k (X_{\mathcal{J}_k},Y_{\mathcal{J}_k})$ defined in \eqref{equ:p_k} satisfies 
\begin{equation}\label{equ:p_k_gaussian}
\hat P_k (X_{\mathcal{J}_k},Y_{\mathcal{J}_k}) =  1 -  \Phi\bigg(   \hat f_k(\mathcal J_k ): = \left[ V_{k}  + \tilde V_{k}\right]^{-1/2}\left[ E_{k} - \tilde E_k \right]  \bigg)  +  O_{\mathbb P }\left(|\mathcal J_k|^{-1/2}\right),
\end{equation}
where $\Phi(\cdot)$ is the cumulative distribution function of the standard normal,
\begin{equation}\label{equ:ev}
\begin{split}
 E_k - \tilde E_k & = 2\sum_{j\in \mathcal J_k}\text{sign}(\tau(X_j)) |Y_j- \mu (X_j) | C_j \  \text{ and } \\ 
 V_k+\tilde V_k &   = \sum_{j\in \mathcal J_k} \left[Y_j-\mu(X_j)\right]^2 [2-C_j^2].
 \end{split}
\end{equation}
\end{theorem}

\begin{proposition}\label{prop:ej}
    In the setup of \Cref{thm:expected_power}, the posterior assignment probability  $e(X_i,Y_i) =\mathbb P\{Z_i = 1 \mid X_i,Y_i\}$ can be expressed using the sigmoid function $\sigma$ as
\begin{equation}\label{equ:ej}
e(X_i,Y_i) = \sigma \left([Y_i-\mu(X_i)]\tau(X_i)/\nu^{2}\right).
\end{equation}
\end{proposition}
\Cref{thm:expected_power} suggests minimizing $\hat P_k(X_{\mathcal{J}_k},Y_{\mathcal{J}_k})$ by maximizing $\hat f_k(\mathcal J_k )$ in \eqref{equ:ev}, which depends on the mean difference $E_k - \tilde E_k$ and the variance sum $V_k + \tilde V_k$ of the observed and randomized statistics in their Gaussian limits.
\Cref{prop:ej} shows that high certainty $e(X_j, Y_j) \approx 1$ arises when $Y_j > \mu(X_j)$ and $\tau(X_j) >0$. To maximize the mean difference, $\mathcal{J}_k$ should include units with positive CATEs and large certainty scores $C_j = |2e(X_j, Y_j) - 1|$. However, including units with large $|Y_j - \mu(X_j)|$ may inflate the variance term $V_k + \tilde V_k$ in \eqref{equ:ev}, thus the trade-off must be considered.

To analyze this trade-off, we relax the discrete optimization problem over $\mathcal{J}_k$ by introducing a continuous vector $\xi \in [0,1]^{|\mathcal{S}_k|}$ that \emph{softly} indicates whether each unit $j \in \mathcal{S}_k$ is included in $\mathcal{J}_k$. The relaxed objective function is defined as
\begin{equation}\label{equ:l_k}
\hat{l}_k(\xi) = \left[ V_k(\xi) + \tilde{V}_k(\xi) \right]^{-1/2} \left[ E_k(\xi) - \tilde{E}_k(\xi) \right],
\end{equation}
where $E_k(\xi)$, $\tilde E_k(\xi)$, $V_k(\xi)$, and $\tilde V_k(\xi)$ follow the same definitions as in \eqref{equ:ev}, except that the sums are taken over all $j \in \mathcal{S}_k$ and weighted by $\xi_j$.

\vspace{5pt}
\begin{proposition}\label{prop:threshold}
Suppose that \( |Y_j - \mu(X_j)| \neq 0 \), and that \( C_j \) are distinct across \( j \in \mathcal{S}_k \). Then the maximizer $\xi^*$ of $\hat l_k(\xi)$ in \eqref{equ:l_k}  takes a threshold form:
\[
\xi_j^* = 
\begin{cases}
1, & \text{if } h(X_j, Y_j) > c \text{ and } \tau(X_j) > 0, \\
c', & \text{if } h(X_j, Y_j) = c \text{ and } \tau(X_j) > 0, \\
0, & \text{if } h(X_j, Y_j) < c \text{ or } \tau(X_j) \le 0,
\end{cases}
\]
where $h(X_j, Y_j) = C_j / \left[ |Y_j - \mu(X_j)| (2 - C_j^2) \right]$, and $c \geq 0$, $c' \in [0,1)$.
\end{proposition}

Most entries of $\xi^*$ are binary, so the relaxed problem closely approximates the original discrete one.
When $\tau(X_j) > 0$, the function $h(X_j, Y_j)$ quantifies unit $j$'s contribution to the mean-variance ratio in \eqref{equ:ev}. In this function, the ratio $C_j / [2 - C_j^2]$ increases with $C_j$.
The remaining term $|Y_j - \mu(X_j)|$ can be approximated by $|\tau(X_j)| / 2$ if we ignore the error $\epsilon_j$ in \eqref{equ:normal_model}.
Using this approximation and direct differentiation, we show in \Cref{sect:derivative} that $h(X_j, Y_j)$ increases with $C_j$ for most values of $C_j$.
Hence, $\xi^*$ can be interpreted as thresholding the certainty score  $C_j$: the $p$-value in \Cref{thm:expected_power} is minimized by including units with positive CATEs and high certainty scores in the inference fold---precisely the allocation strategy used by AdaSplit.
Further details of the AdaSplit algorithm are provided in \Cref{sect:alg}.

\subsection{BaR-learner: CATE estimation with imputed assignments}\label{sect:bar}

Existing CATE estimation methods, such as R-learner \citep{robinson1988root,nie2021quasi}, assume full data availability, with $O_i = (X_i, Y_i, Z_i)$ observed for every unit $i$. Our setting departs from this assumption: units in the inference fold $\mathcal{J} = [n] \setminus \mathcal{I}$ only provide $X_{\mathcal{J}}$ and $Y_{\mathcal{J}}$, while their treatment assignments $Z_{\mathcal{J}}$ are held out for inference. Nevertheless, these units are not arbitrary---by the design of AdaSplit, they tend to have high certainty scores $C_{\mathcal{J}}$, indicating that $Z_{\mathcal{J}}$ can be imputed from $X_{\mathcal{J}}$ and $Y_{\mathcal{J}}$ with high confidence. This imputation can be used to boost the effective sample size and improve the accuracy of the resulting CATE estimates.

Motivated by this, we propose a variant of the popular CATE estimation method R-learner, which we call Bayesian R-learner (BaR-learner). BaR-learner leverages all the covariate and outcome data through its loss:
\begin{equation}\label{equ:br_learner}
          \hat \tau_{\mathcal I }  = \argmin_{\tau } \left\{ \mathcal L_{\text{full}}(O_{\mathcal I}) + \lambda \mathcal L_{\text{imputed}}(X_{\mathcal J},Y_{\mathcal J };\hat e_{\mathcal I})  \right\}, 
\end{equation}
where $\mathcal L_{\text{full}}(O_{\mathcal I})$ is the original loss\footnote{The loss is derived by plugging the expression of $\mu_0$ from \eqref{equ:two_mu} into the outcome model 
\eqref{equ:normal_model}.} 
of R-learner evaluated on the nuisance fold $O_{\mathcal I}$,
\[
\mathcal L_{\text{full}}(O_{\mathcal I})  =  \sum_{i\in \mathcal I } \left\{Y_i - \hat \mu (X_i) - [Z_i-e(X_i)]\tau(X_i) \right\}^2.
\]
To incorporate information from $\mathcal J$, we marginalize out the held-out $Z_{\mathcal J}$ using an estimator $\hat e_{\mathcal I}(x,y)$ of the posterior probability $e(x,y)$. The resulting imputed loss is
\begin{align*}
\mathcal L_{\text{imputed}}(X_{\mathcal J}, Y_{\mathcal J}; \hat e_{\mathcal I}) = & \sum_{j \in \mathcal J} [1 - \hat e_{\mathcal I}(X_j, Y_j)] \left\{ Y_j - \hat \mu(X_j) - [0 - e(X_j)] \tau(X_j) \right\}^2 \\
& + \sum_{j \in \mathcal J} \hat e_{\mathcal I}(X_j, Y_j) \left\{ Y_j - \hat \mu(X_j) - [1 - e(X_j)] \tau(X_j) \right\}^2.
\end{align*}
In practice, one may place greater weight on the full-data loss and downweight the imputed loss. For simplicity, we fix $\lambda = 1$ in \eqref{equ:br_learner} rather than tuning it via cross-validation on $O_{\mathcal I}$. A key challenge is that the estimator $\hat e_{\mathcal I}$ may be biased due to the data-driven selection of the nuisance fold $\mathcal I$. Following the approach of \citet{horvitz1952generalization}, we correct for this bias by re-weighting each observation in $\mathcal I$ using its estimated selection probability given $(X_i, Y_i)$; see \Cref{sect:consistency}, especially \eqref{equ:e_I_hat}, for estimation details and theoretical results.
Like the outcome function $\mu$, the selection probability can be estimated using all $n$ units, and their estimation errors control the bias of $\hat \tau_{\mathcal I}$. Incorporating the imputed loss reduces the variance of $\hat \tau_{\mathcal I}$ to order $n^{-1}$. When both estimators are consistent, the objective in \eqref{equ:br_learner} converges to the population loss $\mathbb E {[Y - \mu(X) - [Z-e(X)]\tau(X) ]^2 }$, yielding a consistent estimator of $\tau$. We also verify this consistency through simulations in \Cref{sect:appendix_exp_consistent}.

In finite samples, suppose most units in $\mathcal{J}$ have high certainty scores $C_j$, such that $\hat e(X_j, Y_j) \approx e(X_j, Y_j) \approx Z_j$. Then the objective in \eqref{equ:br_learner} simplifies to
\[
 \mathcal L_{\text{full}}(O_{\mathcal I}) + \mathcal L_{\text{imputed}}(X_{\mathcal J},Y_{\mathcal J };e ) \approx \mathcal L_{\text{full}}(O_{\mathcal I}) + \mathcal L_{\text{full}}(O_{\mathcal J}) \approx \mathcal L_{\text{full}}(O_{[n]}),
\]
so that $\hat \tau_{\mathcal I}$ closely approximates $\hat \tau_{[n]}$, the minimizer of R-learner’s loss over all data $O_{[n]}$. The splitting strategy in AdaSplit enables accurate imputation to improve estimation.

Like R-leanrer,  BaR-learner can allow any regression models to construct the estimator $\hat \tau_{\mathcal I}$ in \eqref{equ:br_learner}. When $\tau$ is linear,  $\hat \tau_{\mathcal I}$ is obtained by regressing the scaled residual
\begin{equation}\label{equ:residuals}
\hat R_j = R_j(X_j, Y_j, Z_j) := [Y_j - \hat \mu(X_j)] / [Z_j - e(X_j)]
\end{equation}
or its imputed version
$\hat R(X_j,Y_j)  := \hat e_{\mathcal I}(X_j,Y_j)\hat R(X_j,Y_j,1) + [1-\hat e_{\mathcal I }(X_j,Y_j)]\hat R(X_j,Y_j,0)$,
onto the covariates $X_j$ across all $j \in [n]$, as formalized below.
\begin{proposition}\label{prop:r_learner}
Under the Bernoulli$(1/2)$ design, suppose $\tau(x) = x^{\top}\beta$ for some vector $\beta$, and assume that $\hat e_{\mathcal I}(x,y) = e(x,y)$.  Then the estimator  $\hat \tau_{\mathcal I }$ in \eqref{equ:br_learner}  takes the form
\begin{align}\label{eq:solution}
 \hat \tau_{\mathcal I } (x) = x^{\top}\hat \beta_{\mathcal I }  \ \text{ where } \ \hat \beta_{\mathcal I }  = (X_{[n]}^\top X_{[n]})^{-1}\bigg\{ 
    \sum_{i\in \mathcal I }X_i \hat R_i + \sum_{j\in \mathcal J  }X_j \hat R(X_j,Y_j)\bigg\}.
\end{align}
The expected distance between $\hat \beta_{\mathcal I}$ and $\hat \beta_{[n]}$ is given by
\[
 \mathbb E \left[ \| \hat \beta_{\mathcal I } - \hat \beta_{[n]} \|^2 \mid X_{[n]},Y_{[n]} \right] =4 \sum_{j\in [n]\setminus \mathcal I}  X_j^{\top} (X_{[n]}^\top X_{[n]})^{-2}X_j 
\cdot [Y_j- \hat \mu(X_j)]^2  \cdot    [1-C_j^2].
\]
\end{proposition}
\Cref{prop:r_learner} decomposes the impact of excluding unit $j$ from the nuisance fold $\mathcal I$ into three terms.
The first reflects how different $X_j$ is from the dominant spectral modes of $X_{[n]}$, which we use to guide the initialization of AdaSplit later. The second, the squared residual, may increase with $C_j$; we analyze its trade-off with $1 - C_j^2$ below.

\begin{proposition}\label{prop:variance}
Under \Cref{example:1}, it holds that 
\[
1 - C_j^2 = 4\sum_{t=1}^{\infty}  (-1)^{t+1} t \exp \left\{-t |\tau(X_j)[Y_j-\mu(X_j)]| /\nu^{2} \right\}.
\]
\end{proposition}
When $|\tau(X_j)[Y_j - \mu(X_j)]| \geq 2\nu^2$, i.e., when the signal exceeds the noise level in model \eqref{equ:normal_model}, the leading term in the expansion of $1 - C_j^2$, when multiplied by $[Y_j - \hat \mu(X_j)]^2$, decreases with $C_j$. 
This means lower-certainty units reduce the gap between $\hat \beta_{\mathcal I}$ and $\hat \beta_{[n]}$ more when included in the nuisance fold $\mathcal I$.
Echoing \Cref{prop:threshold}, this supports allocating low-certainty units to $\mathcal I$ and reserving high-certainty units for inference.

\subsection{Algorithm of AdaSplit}\label{sect:alg}

We now describe the full procedure of AdaSplit, following the steps in Algorithm~\ref{alg:ass}.

At initialization, we construct an estimator $\hat \mu$ using a regression model fitted to $X_{[n]}$ and  $Y_{[n]}$.
In the function $\texttt{Split}(\mathcal{S}_{[K]}; p)$, we
choose a proportion $p=0.05$ of units from each subgroup $S_{k}$ to form the nuisance fold $\mathcal I,$ and define the remaining units as the inference fold $\mathcal{J}_k$ for every subgroup $k\in [K].$ To keep the algorithm deterministic, we initialize $\mathcal I$ using the units with the largest diversity scores $X_i^{\top} (X_{[n]}^\top X_{[n]})^{-2}X_i$ in \eqref{prop:r_learner}. 

In the function \texttt{Posterior}$(Z_{\mathcal{I}})$, we first apply the R-learner method, i.e., the loss in \eqref{equ:br_learner} with $\lambda = 0$, to construct an initial CATE estimator $\hat \tau^{(0)}$. When using a linear model, as in our experiment, we compute the least-squares solution $\hat \tau^{(0)}$ defined in \eqref{equ:r_} in \Cref{sect:consistency}.
We then estimate the assignment probability $e(x, y)$ based on \eqref{equ:ej}:
\begin{equation}\label{equ:e_model}
\hat e^{(0)}(x,y) = \sigma \left(\hat \tau^{(0)}(x)\left[y-\hat \mu(x)\right]/ [\hat\nu^{(0)}]^2\right),
\end{equation}
where $[\hat\nu^{(0)}]^2$ is an estimator of the variance $\nu^2$ in \eqref{equ:normal_model}, computed using the data $O_{\mathcal I}$. 

\begin{algorithm}[t]
    \caption{Adaptive sample splitting (AdaSplit) for subgroup analysis}\label{alg:ass}
   \vspace{2pt}
\textbf{Input:} Covariates $X_{[n]}$, Outcomes $Y_{[n]}$, Treatment assignments $Z_{[n]}$, \\
  \vspace{2pt}
      \hspace{37pt}    Subgroups $\mathcal S_{[K]}$, Initial proportion $p$, Proportion threshold $\rho$,  \\
        \vspace{2pt}
   \hspace{37pt}        Stopping threshold $\epsilon_l$, Window size $n_0$ \\
   \vspace{5pt}
\textbf{Initialization:}  
$t\gets 1$,\ $\pi_{[K]} \gets 1$,\ $l_{2-n_0},\dotsc l_{n}  \gets 0, \ \mathcal{J} \gets [n]$,\ \ $\hat \mu\gets \mathcal{A}(X_{[n]},Y_{[n]}) $ \\
   \vspace{5pt}
\textbf{Note:}  $X_{[n]},Y_{[n]},$  and $\hat \mu$ are kept implicit below, as they remain fixed throughout. \\
\vspace{5pt}
{\color{blue} Split the sample and create the nuisance estimators}   \\
\vspace{5pt}
$\mathcal{I},\mathcal{J}_{[K]} \gets $\texttt{Split}$(\mathcal{S}_{[K]}; p)$,  
    $\hat{e}^{(0)}  \gets$ \texttt{Posterior}$(Z_{\mathcal{I}} )$
\vspace{10pt}          \\ 
\While{$\max\{l_{t-n_0+1},\dots, l_t\}\leq \one\{t\leq n_0\} +  \epsilon_l $ \text{ and }  $\max\{\pi_{[K]}\} \leq \rho$ 
}{
\vspace{7pt}
{\color{blue} 1. Select units from subgroups that meet the threshold $\rho$ }   \\
\vspace{2pt}
$(j^*,k^*) \gets \argmin_{j\in \mathcal J_k,\ k\in [K]:\pi_k \geq \rho} \text{sign}\big(\hat \tau^{(t-1)}(X_j) \big)\big|2\hat{e}^{(t-1)}(X_j,Y_j)-1\big|$  \\
\vspace{5pt}
$\mathcal{I} \gets \mathcal{I}\cup \{j^*\}$, \ 
$\mathcal{J}_{k^*} \gets \mathcal{J}_{k^*}\setminus \{j^*\}$, \ $\pi_{k^*}\gets |\mathcal{J}_{k^*}|/|\mathcal{S}_{k^*}|$ \\
\vspace{7pt}
{\color{blue} 2. Update the nuisance estimators } \\
\vspace{5pt}
 $\hat{e}^{(t)}  \gets$ \texttt{Posterior}$(Z_{\mathcal{I}})$,\ 
   $\hat{\tau}^{(t)}  \gets$ \texttt{BaR-learner}$(Z_{\mathcal{I}}; \hat{e}^{(t)})$
   \\
\vspace{7pt}
{\color{blue} 3. Check convergence of loss} \\
\vspace{3pt}
 $l_t \gets 1 - R^2(\hat{\tau}^{(t)},\hat{\tau}^{(t-1)};X_{\mathcal{J}})$, \ $t \gets t +1$ 
 \vspace{5pt}
}
Remove units with the most negative $\hat{\tau}^{(t)}(X_j)$ from $\mathcal{J}_k$ until $\pi_k <\rho $ \\
\vspace{5pt}
$\mathcal I \gets [n]\setminus \bigcup_{k\in [K]} \mathcal J_k$   ,\ 
 $\hat{e}_{\mathcal I}  \gets$ \texttt{Posterior}$(Z_{\mathcal{I}})$,\ 
  $\hat{\tau}_{\mathcal I}  \gets$ \texttt{BaR-learner}$(Z_{\mathcal{I}}; \hat{e}_{\mathcal I})$
   \\
\vspace{5pt}
\textbf{Output:} $p$-values $\hat P_k(O_{\mathcal{J}_k})$ in \eqref{equ:pk} for all $k \in [K]$
\vspace{2pt}
\end{algorithm}
At iteration $t = 1, 2, \dotsc$, we implement a unit selection process as follows.
\begin{itemize}
    \item 
In Step 1, we update the nuisance fold $\mathcal I$ with the unit  $j^*$ that minimizes
\begin{equation}\label{equ:psi}
\text{sign}\big(\hat \tau^{(t-1)}(X_j) \big)\hat{C}_j^{(t-1)} := \text{sign}\big(\hat \tau^{(t-1)}(X_j) \big)\big|2\hat{e}^{(t-1)}(X_j,Y_j)-1\big|,
\end{equation}
across all the subgroups with the inference proportion $\pi_k  = |\mathcal J_k|/|\mathcal S_k| \geq \rho$, e.g., $\rho = 0.5$. 
Minimizing \eqref{equ:psi} prioritizes units with negative estimated CATEs. Once these are exhausted, it tends to select those with positive estimated CATEs and low estimated certainty scores $\hat{C}_j^{(t-1)}$.
After identifying $j^*$, we remove it from the inference fold $\mathcal{J}_{k^*}$ it belongs to, and update $\pi_{k^*} = |\mathcal{J}_{k^*}|/|\mathcal{S}_{k^*}|$ accordingly.
\item  In Step 2,  \texttt{Posterior}$(Z_{\mathcal{I}}$) computes a new posterior assignment probability estimator $\hat{e}^{(t)}$ as in \eqref{equ:e_model}, using a new R-learner estimator that corrects for selection bias via inverse probability weighting, as described in \eqref{equ:r_weighted} in \Cref{sect:consistency}. We then apply the function \texttt{BaR-learner}$(Z_{\mathcal{I}}; \hat{e}^{(t)})$, i.e., the loss function in \eqref{equ:br_learner} with $\lambda = 1$,
on the updated inference fold $\mathcal I$ to construct a new estimator,
$\hat{\tau}^{(t)}(x) = x^{\top}\hat \beta^{(t)},$
where $\hat \beta^{(t)}$ is defined analogously to $\beta_{\mathcal I}$ in \eqref{eq:solution}, except that the  assignment probability $e(x,y)$ is replaced by its estimate $\hat e^{(t)}(x,y)$.\footnote{ If the regression model used in Step 2 is computationally expensive to train, we may update it every few iterations (e.g., every 20) and adjust the convergence threshold in Step 3 accordingly.}
\item 
In Step 3, after the first $n_0$ iterations, e.g., $n_0 = 50$, we check the convergence of the CATE estimates for terminating the unit selection process. We compute $1 - R^2$ (one minus the coefficient of determination) to assess the change in predictions of $\hat{\tau}^{(t)}$ relative to $\hat{\tau}^{(t-1)}$ on the covariates $X_{\mathcal J}$. We terminate the selection process if this change remains below a threshold $\epsilon_l = 0.01$ for the last $n_0$ iterations, or if the inference proportion $\pi_k < \rho$ for all subgroups $k \in [K]$. This stopping rule ensures that the CATE estimates have converged and that the inference proportions across all subgroups are at least $\rho$, which matches the proportion used in random sample splitting in all our experiments.
\end{itemize}

According to the p-value in \Cref{thm:expected_power} and the solution in \Cref{prop:threshold}, we remove from $\mathcal J_k$ the units with the most negative estimated CATEs, if such units exist and $\pi_k \geq \rho$ after termination. These units are added into the nuisance fold $\mathcal{I}$ to
update $\hat e_{\mathcal I}$ as in Step 2 and compute the final estimator $\hat \tau_{\mathcal{I}}$ using \texttt{BaR-learner}$(Z_{\mathcal{I}}; \hat{e}_{\mathcal I})$. This estimator is then applied to compute the $p$-values $\hat P_k(O_{\mathcal{J}_k})$ in \eqref{equ:pk}.

\vspace{-3pt}
\subsection{Validity of AdaSplit}\label{sect:valid}

We now show that AdaSplit yields valid \( p \)-values for both single and multiple hypothesis testing. We let \( \mathcal{K}_0 = \{k \in [K] : H_{0,k} \text{ in } \eqref{equ:null} \text{ is true} \} \) denote the set of null groups, and \( \mathcal{K}_1 = [K] \setminus \mathcal{K}_0 \) as the set of non-nulls. For $z\in \{0,1\}$, we let  \( \mathcal{J}_{\mathcal{K}_z} = \bigcup_{k \in \mathcal{K}_z} \mathcal{J}_k \).

In Algorithm~\ref{alg:ass}, the nuisance fold \( \mathcal{I} = \mathcal{I}(O_{[n]}) \) is constructed iteratively by selecting the unit that minimizes the objective function \eqref{equ:psi} at each step. This relies only on the treatment assignments revealed up to the current iteration. If the final nuisance fold \( \mathcal{I}(O_{[n]}) = I \), modifying the assignments of any units outside \( I \) would yield the same nuisance fold. This invariance property leads to validity, as formalized below.

\begin{theorem}\label{thm:validity}
Suppose Assumptions~\ref{ass:randomization} and~\ref{ass:consistency} hold. Then, for any \( k \in \mathcal{K}_0 \), the \( p \)-value \( \hat{P}_k(O_{\mathcal{J}_k}) \) returned by Algorithm~\ref{alg:ass} satisfies
\begin{equation}\label{equ:validity}
    \mathbb{P} \left\{ \hat{P}_k(O_{\mathcal{J}_k}) \leq \alpha \,\middle|\,
    X_{[n]}, Y_{[n]}, Z_{\mathcal{I} \cup \mathcal{J}_{\mathcal{K}_1}}, \mathcal{I}(O_{[n]}) \right\} \leq \alpha.
\end{equation}
Furthermore, the null \( p \)-values \( \hat{P}_k(O_{\mathcal{J}_k}) \), for \( k \in \mathcal{K}_0 \), are jointly independent conditional on the same random variables as in \eqref{equ:validity}.
\end{theorem}

The \( p \)-value validity in \eqref{equ:validity} guarantees control of the type I error rate for testing each subgroup null hypothesis  individually. However, when rejecting multiple nulls to claim that the treatment has a significant effect in several subgroups, it becomes important to control the family-wise error rate (FWER), defined as the probability of making one or more false rejections. A principled way to achieve strong FWER control is through the closed testing procedure \citep{marcus1976closed}.

In this procedure, each null \( H_{0,k} \) is tested by evaluating the intersection nulls \( H_{0,\mathcal{K}} = \bigcap_{j \in \mathcal{K}} H_{0,j} \) for all subsets \( \mathcal{K} \subseteq [K] \) that contain \( k \). The null \( H_{0,k} \) is rejected if all such intersection nulls are rejected. Specifically, the procedure conducts a global test to generate a $p$-value $\hat{P}_\mathcal{K}$ 
for each intersection null $H_{0,\mathcal{K}}$, and defines the rejection set:
\[
\mathcal{R} = \left\{ k \in [K] : \hat{P}_\mathcal{K} \leq \alpha \text{ for all } \mathcal{K} \subseteq [K] \text{ with } k \in \mathcal{K} \right\}.
\]

\begin{theorem}[\citep{marcus1976closed}]
In the setup of \Cref{thm:validity}, the rejection set \( \mathcal{R} \) controls the family-wise error rate (FWER) at level \( \alpha \):
\[
\mathbb{P} \left\{ \mathcal{R} \cap \mathcal{K}_0 \neq \emptyset \,\middle|\,
X_{[n]}, Y_{[n]}, Z_{\mathcal{I} \cup \mathcal{J}_{\mathcal{K}_1}}, \mathcal{I}(O_{[n]}) \right\} \leq \alpha.
\]
\end{theorem}

The FWER control holds under arbitrary dependence among the \( p \)-values. Nevertheless, the independence of our \( p \)-values in \Cref{thm:validity} enables us to apply powerful global tests for the intersection nulls, e.g., Fisher’s method \citep{fisher1928statistical}, in place of commonly used but conservative procedures like the Holm method \citep{holm1979simple}.

\section{Experiment}\label{sect:exp}

We evaluate the randomization tests produced by Algorithm \ref{alg:ass} and compare them against two baselines.\footnote{Code to reproduce the simulation studies and real data analysis is available at
\url{https://github.com/ZijunGao/AdaSplit}.} The first is the vanilla subgroup-based randomization test described in \Cref{sect:vanilla}, which uses all units in the inference fold to compute the difference-in-means statistic $T_{\text{DM}}$.
The second baseline applies random sample splitting, allocating one fold for CATE estimation and the other for randomization tests. The reference distributions in all $p$-values are computed using Monte-Carlo with
1,000 draws of randomized treatment assignments. We refer to the first baseline as RT, the second as RT (RandomSplit), and our method as RT (AdaSplit).

RT (RandomSplit) and RT (AdaSplit) use the AIPW statistic, where we estimate $\mu$ and $\tau$ using linear regression; the same experiments are repeated in \Cref{sect:appendix_exp_xgboost} with $\mu$ estimated via XGBoost. By default, RT (RandomSplit) splits the units evenly, i.e., the proportion of units in the nuisance fold is $\rho = 0.5$. RT (AdaSplit) also sets the maximum proportion of units in the nuisance fold $\mathcal I$ to $\rho$, and initializes $\mathcal I$ in Algorithm \ref{alg:ass} using the 5\% of units with the highest diversity scores.
We repeat the experiments with various proportions in \Cref{sect:appendix_exp_results}.

\subsection{Experiments on synthetic data}

\subsubsection{Setup} \label{sect:exp_setup}

Our experiment here has three settings.  In the default setting, we set the number of units to $n=500$. Every unit $i\in [n]$ has five covariates: $X_1,X_2$ and $X_3$ are independently drawn from a uniform distribution on $[-0.5, 0.5]$, while $X_4$ and $X_5$ are independent Bernoulli  taking values in $\{-0.5, 0.5\}$ with $\mathbb{P}(X_4 = 0.5) = 0.25$ and $\mathbb{P}(X_5 = 0.5) = 0.75$.
The treatment assignments and outcomes are generated following the setup of \Cref{example:1}.
Specifically, we fix the noise variance at $\nu^2 = 1$, 
define both $\mu_0$ and $\mu_1$ as linear functions of the covariates.
The CATE is given by
\begin{equation}\label{equ:tau_synthetic}
\tau(x) = 0.5 + \sum_{i=1}^5 x_i.
\end{equation}
Building on the default, we introduce two other settings: (1) \quotes{Larger sample size}, where the sample size increases from $n=500$ to $n = 1000$, and (2) \quotes{Increased noise level}, where the noise variance increases from $\nu^2 = 1$ to $\nu^2 = 2$. 

We create five subgroups $(K=5)$ based on the $0.2$, $0.4$, $\ldots$, and $0.8$-quantiles of $X_1$; for example, the first subgroup contains units $i$ with
$X_{i,1}$ below the $0.2$-quantile.

\subsubsection{BaR-learner v.s. R-learner}\label{sect:r_bar}

We first compare our proposed BaR-learner with the original R-learner in estimating the true CATE $\tau$ across the three simulation settings mentioned above.
To assess the estimators $\hat{\tau}$, we compute the out-of-sample $R^2$ using a hold-out dataset of size $10^4$:
\begin{align*}
    R^2 = 1 - \sum_{i'=1}^{10^4} \left(\tau(X_{i'}) - \hat{\tau}(X_{i'})\right)^2/\sum_{i'=1}^{10^4} \left(\tau(X_{i'}) - \bar{\tau}\right)^2, \   
\end{align*}
where $\bar{\tau} = 10^{-4}\sum_{i'=1}^{10^4}\tau(X_{i'})$.
The larger $R^2$ is, the more accurate $\hat{\tau}$ is. 
\Cref{tab:R2} shows that  BaR-learner produces significantly more accurate CATE estimates than R-learner. 
The poor performance of R-learner can be attributed to the efficiency loss due to discarding all units in the inference fold. In contrast, by imputing treatment assignments in the inference fold, BaR-learner is able to leverage all units' covariates and outcomes for estimation.  We observe consistent benefits from this imputation strategy across all settings, including the \quotes{Increased noise level} scenario, where treatment assignments $Z_i$ are harder to predict from $X_i$ and $Y_i$.

\begin{table}[tbh]
\centering
\caption{Out-of-sample $R^2$ of the CATE estimators from BaR-learner and R-learner across three simulation settings. Results are obtained from 100 independent runs.}
\label{tab:R2}
\begin{tabular}{c|ccc}
\toprule
Method & Default & Larger sample size & Increased noise level \\
\midrule
R-learner   & $0.49$ (0.06) & $-1.58$ (0.34) & $-1.56$  (0.47) \\
BaR-learner & 0.79 (0.01) & 0.43 (0.03) & 0.43 (0.06) \\
\bottomrule
\end{tabular}
\end{table}

\subsubsection{Single hypothesis testing}

\textbf{Validity}. To test the validity of the randomization tests, we set $\tau(x) = 0$ in \eqref{equ:tau_synthetic}, turning all subgroups into null. In \Cref{tab:typeIerror}, all methods control their type I errors at the nominal level $0.2$ across all subgroups. This result is expected for RT and RT (RandomSplit) since they do not select units in a data-driven way. More importantly, it confirms that our adaptive procedure preserves the validity of randomization tests.

\begin{table}[h]
\centering
\caption{Type I errors of the randomization tests at the nominal level $\alpha = 0.2$ for five null groups (G1–G5). 
The results are aggregated over $200$ trials.}
\label{tab:typeIerror}
\begin{tabular}{c|ccc}
\toprule
Method & RT & RT (RandomSplit) & RT (AdaSplit) \\
\midrule
G1 & 0.180 (0.038) & 0.160 (0.037) & 0.180 (0.038) \\
G2 & 0.170 (0.038) & 0.215 (0.041) & 0.225 (0.042) \\
G3 & 0.190 (0.039) & 0.210 (0.041) & 0.220 (0.041) \\
G4 & 0.240 (0.043) & 0.195 (0.040) & 0.200 (0.040) \\
G5 & 0.165 (0.037)  & 0.195 (0.040) & 0.170 (0.038) \\
\bottomrule
\end{tabular}
\vspace{5pt}
\end{table}

\begin{figure}[t]
    \begin{minipage}{0.27\textwidth}
        \centering
        \includegraphics[clip, trim = 0.2cm 0cm 5cm 0cm, height = 0.9\textwidth]{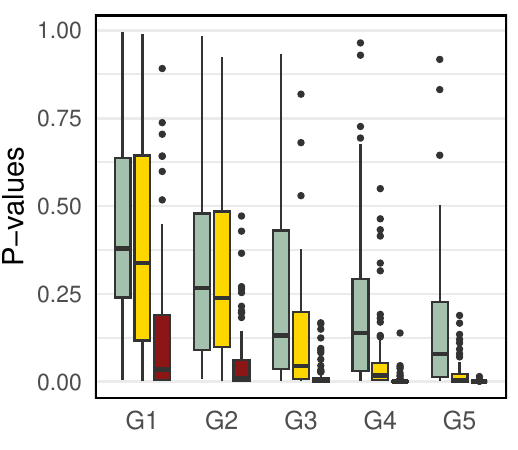}
        \subcaption*{ (a) Default setting.}
    \end{minipage}
    \hspace{0cm}
     \begin{minipage}{0.27\textwidth}
        \centering
        \includegraphics[clip, trim = 0cm 0cm 5cm 0cm, height = 0.9\textwidth]{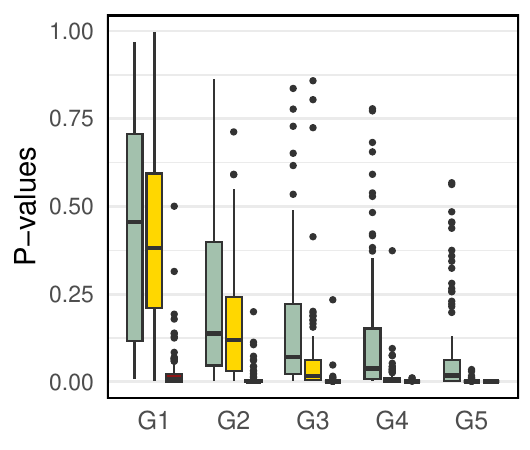}
        \subcaption*{\quad (b) Large sample size.}
    \end{minipage}
    \hspace{0cm}
         \begin{minipage}{0.27\textwidth}
        \centering
        \includegraphics[clip, trim = 0cm 0cm 0cm 0cm, height = 0.9\textwidth]{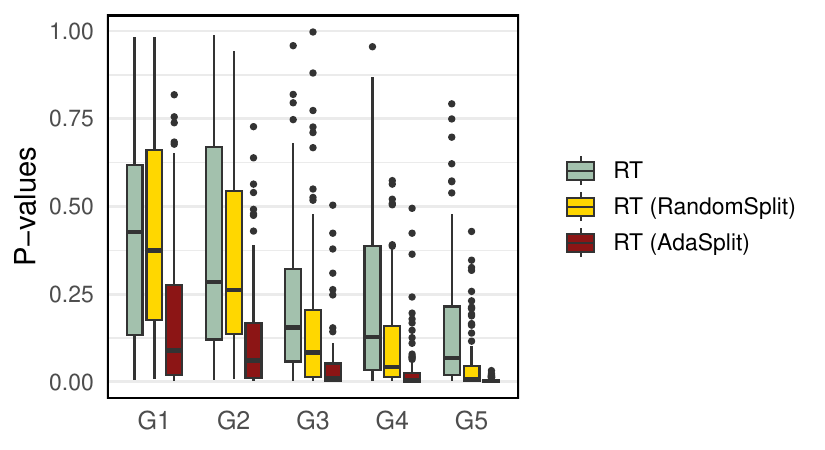}
        \subcaption*{ ~~(c) Increased noise level.}
    \end{minipage}
    \vspace{5pt}
    \caption{Boxplots of subgroup $p$-values generated by RT, RT (RandomSplit), and RT (AdaSplit) across three different settings.
The results are aggregated over $100$ trials.
}
    \label{fig:group.p.val}
\end{figure}

\begin{figure}[h]
\vspace{10pt}
    \centering
     \begin{minipage}{0.35\textwidth}
        \centering
        \includegraphics[clip, trim = 0cm 0cm 0cm 0cm, height = 1\textwidth]{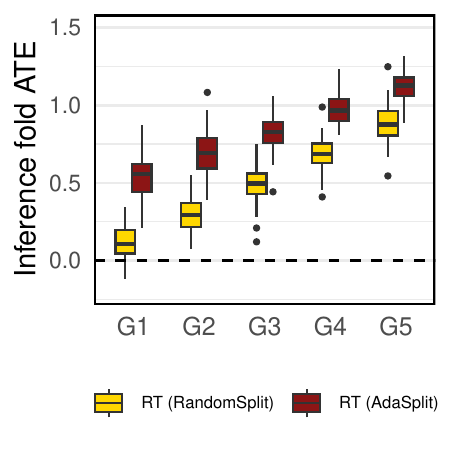}
        \subcaption*{(a) Subgroup ATEs.}
    \end{minipage}
    \hspace{1cm}
    \begin{minipage}{0.35\textwidth}
        \centering
        \includegraphics[clip, trim = 0cm 0cm 0cm 0cm, height = 1\textwidth]{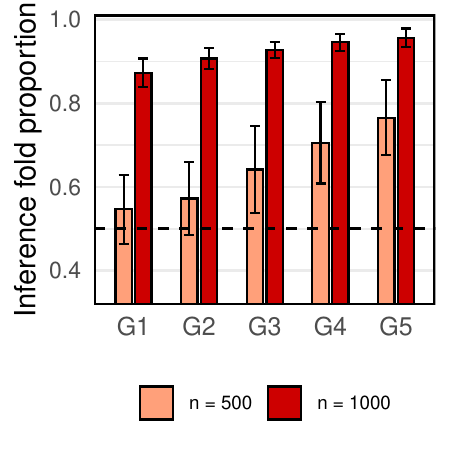}
        \subcaption*{(b) Inference fold proportions.}
    \end{minipage}
\caption{Subgroup ATEs and inference fold proportions
in RT (RandomSplit) and RT (AdaSplit) in the default setting (and the setting with $n=1000$ in the right panel).
}
\label{fig:inference.fold}
\end{figure}

\textbf{Power}.  In \Cref{fig:group.p.val}, we compare the power (i.e., $p$-values) of the randomization tests in the three settings mentioned above.
We first observe that RT is generally less powerful than the other two methods. This indicates that the AIPW statistics using the estimators $\hat{\mu}$ and $\hat{\tau}$ can substantially improve power over the simple difference-in-means statistic.
That said, comparing panel (a) with panel (c) reveals that this advantage diminishes when the estimators become less accurate due to increased noise.

\ART consistently achieves smaller $p$-values than the two non-adaptive methods across all settings. Two key observations in \Cref{fig:inference.fold} help explain this improvement. 
\begin{itemize}
    \item Panel (a)  shows that the average treatment effects (ATEs) (i.e., the observed test statistics) in the inference fold adaptively selected in \ART are larger than those in RT or RT~(RandomSplit). This naturally leads to smaller $p$-values, assuming the reference distributions are similar across methods. 
    \item Panel (b) compares the inference proportions\footnote{The inference proportions reported here correspond to those before the final step of our algorithm, which excludes units with negative CATE estimates.} of \ART at sample sizes $n = 500$ and $1000$. 
We observe that \ART can reserve more than 50\% of the units for inference as the CATE estimator converges early. Moreover, for a fixed sample size $n$, the inference proportions increase from Group G1 to G5, which means the proportion tends to be larger in subgroups with larger CATEs, as defined in the experimental setup in \Cref{sect:exp_setup}. Units with larger CATEs typically have estimates of $e(X_i, Y_i)$ closer to 0 or 1, making them less likely to be assigned to the nuisance fold.
\end{itemize}

\subsubsection{Multiple hypotheses testing}

The previous section evaluates the methods based on subgroup $p$-values. Here, we assess their performance in the context of multiple testing by applying Fisher's method to their $p$-values, controlling the family-wise error rate (FWER) at level $q = 0.2$. \Cref{tab:FWER} shows that the realized FWERs for all methods remain below 0.2. \ART is consistently more powerful than the other methods, as shown in \Cref{tab:power}.

\begin{table}[h]
\centering
\caption{Realized FWERs of RT, \RRT and \ART in the null setting with $\tau(x) = 0$ in \eqref{equ:tau_synthetic} and  $q=0.2$. The results are aggregated over $200$ repeats.}
\vspace{5pt}
\label{tab:FWER}
\begin{tabular}{c|ccc}
\toprule
Method & RT & \RRT & \ART \\
\midrule
Null & 0.105 (0.022) & 0.110 (0.022) & 0.135 (0.024) \\
\bottomrule
\end{tabular} 
\end{table}

\vspace{10pt}

\begin{table}[h]
\centering
\caption{Realized powers of RT, \RRT and \ART across three simulation settings. Results are aggregated over $100$ repeats per setting.}
\label{tab:power}
\begin{tabular}{c|ccc}
\toprule
Method & RT & \RRT & \ART \\
\midrule
Default setting           & 0.298 (0.026) & 0.590 (0.024) & 0.930 (0.012) \\
Larger sample size & 0.496 (0.028) & 0.728 (0.017) & 0.994 (0.004) \\
Increased noise level     & 0.288 (0.026) & 0.500 (0.026) & 0.854 (0.017) \\
\bottomrule
\end{tabular}
\end{table}

\subsection{Experiments on real data}

\begin{figure}[t]
\vspace{5pt}
    \centering
    \begin{minipage}{0.31\textwidth}
        \centering
        \includegraphics[clip, trim = 0cm 0cm 3cm 0.65cm, height = 1.3\textwidth]{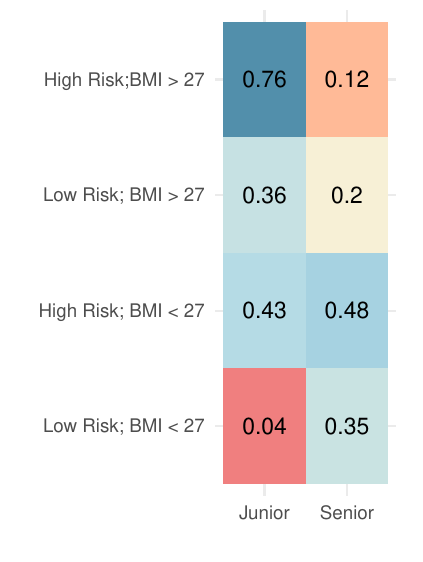}
        \subcaption*{\qquad\qquad\qquad (a) RT.}
    \end{minipage}
    \hspace{0.5cm}
     \begin{minipage}{0.31\textwidth}
        \centering
        \includegraphics[clip, trim = 3.5cm 0cm 3cm 0.65cm, height = 1.3\textwidth]{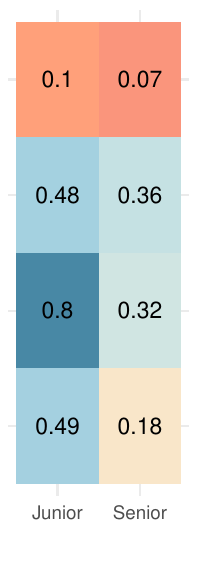}
        \subcaption*{(b) RT (RandomSplit).}
    \end{minipage}
         \begin{minipage}{0.31\textwidth}
        \centering
        \includegraphics[clip, trim = 3.7cm 0cm 1.45cm 0.65cm, height = 1.3\textwidth]{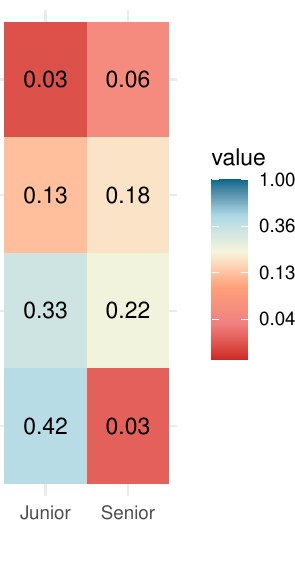}
        \subcaption*{(c) RT (AdaSplit). \qquad\qquad}
    \end{minipage}
    \vspace{3pt}
\caption{Heatmaps of subgroup $p$-values in the SPRINT dataset.
}
\label{fig:group.p.val.SPRINT}
\end{figure}

We apply our method to the Systolic Blood Pressure Intervention Trial (SPRINT) dataset \citep{Wright2016ART,sprint2016systolic, gao2021assessment}, which evaluates whether an intensive systolic blood pressure treatment reduces the risk of cardiovascular disease (CVD). The primary outcome is a binary indicator of whether a major CVD event occurred. 
To ensure larger outcomes indicate better health,  we recode the outcome: 1 indicates no CVD event and 0 indicates an event occurred. The trial uses a Bernoulli design, where every individual's treatment assignment is independently generated from a Bernoulli distribution with probability 1/2.
The dataset has 18 covariates, including demographic information (e.g., age) and baseline clinical measurements (e.g., body mass index, BMI). We retain 8,746 individuals for analysis after removing those with missing data.

We partition the data into $2^3 = 8$ subgroups based on three covariates: age (senior if $>70$, junior if $\le 70$), body mass index (BMI high if $>27$, low if $\le 27$), and 10-year CVD risk (high if ``RISK10YRS'' $> 18$, low if $\le 18$).
We set the FWER level at $0.2$.

We compare our method against the same baselines used in the previous section, keeping each method’s setup unchanged except for estimating the outcome function $\mu$ using XGBoost.  
The heatmap in \Cref{fig:group.p.val.SPRINT} shows that, while the other methods fail to reject any subgroup null hypotheses, \ART yields smaller $p$-values and successfully rejects three subgroups, those in the top row and the bottom right. This result suggests that individuals with both high risk scores and BMI $>27$, as well as older individuals with low risk scores and BMI $<27$, are more likely to benefit from the treatment. The strong effect in the latter group seems unexpected, given their relatively good health. This may be due to limited power for detecting effects in other senior subgroups with higher risk scores or BMI.

\section{Discussion}\label{sect:discussion}

This paper introduced AdaSplit, an adaptive sample splitting method for constructing valid and powerful randomization tests. Our key observation is that, when dividing a sample between estimation and inference, individual units contribute differently to each task. By adaptively allocating units based on these contributions, AdaSplit can achieve better performance than random sample splitting.
A natural question is how broadly the idea of AdaSplit can be applied beyond the current context, as random sample splitting remains a common strategy in many problems that require both estimation and inference. We explore this question below.

\textbf{Other experimental designs.} While we focus on Bernoulli designs and pre-specified subgroups in the main text, AdaSplit can be extended to more complex experimental designs. This requires two adaptations: (1) re-deriving the Gaussian approximation of the p-value in \Cref{thm:expected_power}, which underlies the objective function in \eqref{equ:psi} used for selecting the nuisance fold; and (2) recalculating the posterior assignment probability used in BaR-learner. These modifications could broaden the applicability of AdaSplit to settings such as stratified or cluster-randomized trials, which are widely used in practice but remain underexplored in the context of adaptive inference.

\textbf{Data-adaptive subgroups.}  In the absence of pre-specified subgroups, tree-based CATE estimators \citep{athey2016recursive,hahn2020bayesian} can be integrated into AdaSplit to partition the covariate space into subgroups with different treatment effects. When such methods identify subgroups with large treatment effects, applying randomization tests to these subgroups can yield high power. Crucially, the unit allocation strategy in AdaSplit does not depend on subgroup membership, allowing subgroup definitions to evolve across iterations. This flexibility makes AdaSplit well-suited for applications such as targeted marketing or policy evaluation, where meaningful subgroups are often discovered from data rather than defined a priori.

\textbf{Beyond treatment effects.} 
An interesting direction for future work is to extend the core idea of AdaSplit to other statistical tasks that involve both estimation and inference.
For instance, in change-point detection, one might first fit a parametric model to the data stream over time, and then permute observations on either side of the estimated change-point to compute a valid $p$-value. In this setting, observations near the change-point are more informative for localization, while those farther away contribute more to model fitting.
Similarly, in conditional independence testing, e.g., testing whether $X_j \perp Y \mid X_{-j}$, observations where $X_{i,j}$ is highly correlated with $Y_i$ are more informative for detecting dependence. Following the idea of BaR-learner, such $X_{i,j}$ can be imputed from $(X_{i,-j}, Y_i)$, and the imputed values can be used in fitting a regression model to reduce the variance of the test statistic. 
In addition,  \citet{small2024protocols} highlights the challenge of optimally allocating data between the design and analysis phases of observational studies, which could potentially be addressed by AdaSplit.

\vspace{15pt}

\section{Acknowledgements}

Y.Z was supported by the Office of Naval Research under Grant No. N00014-24-1-2305 and the National Science Foundation under Grant No. DMS2032014. 

\bibliographystyle{plainnat}
\bibliography{references}
 
\clearpage

\appendix
\appendixpage

\section{Additional details on BaR-learner}\label{sect:consistency}

Suppose the design is \text{Bern(1/2)} and $\tau$ is a linear function.
We first show that applying R-learner to the selected nuisance fold $\mathcal I = \mathcal I(O_{[n]})$ may yield a inconsistent CATE estimator. 
Let $B_i = \one\{i\in \mathcal I\}$ indicate whether unit $i$ is included in $\mathcal I $. Let $\Sigma_{\mathcal I} = n_{\mathcal I}^{-1}\sum_{i=1}^{n}B_i X_i X_i^{\top} $ denote the sample covariance matrix based on $\mathcal I $,
where $n_{\mathcal I} = |\mathcal I|=  \sum_{i=1}^{n}B_i$. The full sample covariance matrix $\Sigma_{[n]}$ is defined analogously.

When $\hat \mu$ is a consistent estimator of $\mu$, the residuals $\hat R_i$ defined in \eqref{equ:residuals} satisfy:
\[
\hat R_i- \left(X_i^{\top}\beta +  \frac{\epsilon_i}{Z_i-e(X_i)}\right) \overset{p}{\rightarrow} 0.
\]
Define the least-squares estimator $\hat \tau^{\text{OLS}}(x) = x^{\top} \hat \beta_{\mathcal I}^{\text{OLS}} $ from R-learner as
\begin{equation}\label{equ:r_}
 \hat \beta_{\mathcal I}^{\text{OLS}}  = \Sigma_{\mathcal I}^{-1}\phi_{\mathcal I} \  \text{ with } \  \phi_{\mathcal I} = n_{\mathcal I}^{-1} \sum_{i=1}^{n}B_iX_i\hat R_i.
\end{equation}
Let $U_i :=  \mathbb E[X_i \epsilon_i \mid X_i, B_i]$. 
Due to the selection bias in $\mathcal I$, 
\begin{equation}\label{equ:asymptotic_bias}
\begin{split}
\hat \beta_{\mathcal I}^{\text{OLS}}  - \beta  
 & = \Sigma_{\mathcal I}^{-1} \frac{1}{n_{\mathcal I }}\sum_{i=1}^{n}  \bigg\{ B_i X_i U_i + B_i X_i (\epsilon_i-U_i)    + \frac{\mu(X_i) - \hat \mu (X_i) }{1-e(X_i)} \bigg\} \\
& \overset{p}{\rightarrow} (\mathbb E [ X_i^{\top} X_i \mid B_i=1])^{-1} \mathbb E[X_i \epsilon_i \mid B_i=1].
\end{split}
\end{equation}
In contrast, without conditioning on $B_i$, we have 
$\mathbb E[X_i \epsilon_i \mid X_i] =  X_i \mathbb E[\epsilon_i \mid X_i] = 0$, which means the estimator based on a randomly sampled nuisance fold is consistent. 

To correct for this bias, we draw on the Horvitz-Thompson estimator \citep{horvitz1952generalization} and re-weight observations in $\mathcal I$ by their selection probability  $p(X_i) := \mathbb P\{ B_i = 1\mid X_i,Y_i  \}$. This leads to a new estimator $\hat \tau_{\mathcal I,\hat p}(x) = x^{\top} \hat \beta_{\mathcal I,\hat p }$, where
\begin{equation}\label{equ:r_weighted}
\hat \beta_{\mathcal I,\hat p } = \Sigma_{\mathcal I,\hat p}^{-1}\phi_{\mathcal I,\hat p} :=  \left(  n^{-1}\sum_{i=1}^{n}\frac{B_i}{\hat p(X_i,Y_i)} X_i X_i^{\top}\right)^{-1} \left(  n^{-1}\sum_{i=1}^{n}\frac{B_i}{\hat p(X_i,Y_i)} X_i \hat  R_i\right).
\end{equation}
The estimator $\hat p$ can be obtained by fitting a classifier on $(X_i,Y_i,B_i),i\in [n]$, or computing a local average as in \citet{nadaraya1964estimating,watson1964smooth}:
\[
\hat p(X_i,Y_i)  = \frac{\sum_{j=1}^{n}\kappa([X_j,Y_j],[X_i,Y_i])B_j}{\sum_{j'=1}^{n}\kappa([X_j',Y_j'],[X_i,Y_i])},
\]
where $\kappa$ is a positive similarity measure between its two arguments, e.g., $\kappa = 1$ if $[X_j, Y_j]$ is among the 10 nearest neighbors of $[X_i, Y_i]$.

\begin{assumption}\label{ass:6}
The covariates and outcomes $(X_i, Y_i)$, the outcome function $\mu$, and its estimator $\hat \mu$ are all $\ell_2$-bounded. Moreover, the sample covariance matrices $\Sigma_{\mathcal I, \hat p}$ and $\Sigma_{\mathcal I}$ are almost surely positive definite for any nuisance fold $\mathcal I \subseteq [n]$, with all their eigenvalues bounded below by some constant $c_{\Sigma} > 0$.
\end{assumption}

\begin{assumption}\label{ass:7}
For all $i \in [n]$, $(\epsilon_i, Z_i) \independent B_i \mid X_i, Y_i$.
\end{assumption}
\begin{assumption}\label{ass:8}
There exists $c_p > 0$ such that $p(x, y) \in (c_p, 1]$ for all $(x, y) \in \mathcal{X} \times \mathcal{Y}$.
\end{assumption}

\begin{proposition}\label{prop:p_x}
Under \Cref{ass:6,ass:7,ass:8}, the bias of $\hat \beta_{\mathcal I,\hat p }$ satisfies
\[
\mathbb E \big[  \hat \beta_{\mathcal I,\hat p } \big] - \beta  \lesssim  \mathbb E \big[ | \hat  \mu(X_i) - \mu(X_i) |   \big] + \mathbb E \big[ | \hat  p(X_i,Y_i) - p(X_i,Y_i) |   \big] + O(1/\sqrt{n}).
\]
\end{proposition}
\Cref{prop:p_x} shows that 
$\hat \tau_{\mathcal I,\hat p}(x) = x^{\top} \hat \beta_{\mathcal I,\hat p }$ is an consistent estimator if the estimators $\hat \mu$ and $\hat  p$ are consistent. However, the variance of $\hat \beta_{\mathcal I,\hat p }$ may be inflated by the inverse probability weights  $\hat p^{-1}(X_i,Y_i)$, especially when some $\hat p(X_i,Y_i)$ are close to zero.

We next define the BaR-learner estimator, which uses inverse probability weights indirectly.
We first construct a posterior probability estimator following  \eqref{equ:ej}:
\begin{equation}\label{equ:e_I_hat}
\hat e_{\mathcal I}(x,y) =  \sigma \left([y-\hat \mu(x)]\hat \tau_{\mathcal I,\hat p}(x)/\hat \nu_{\mathcal I,\hat p}^{2}\right),
\end{equation}
where $\hat \nu_{\mathcal I,\hat p}^2 = n^{-1}\sum_{i=1}\hat p^{-1} (X_i,Y_i)B_i\big[Y_i - \hat \mu(X_i) - Z_i\hat \tau_{\mathcal I,\hat p}(X_i) \big]^2$.

The BaR-learner estimator $ \hat \tau_{\mathcal I} (x) = x^{\top}\hat \beta_{\mathcal I } $ is obtained by minimizing the loss in \eqref{eq:solution}:
\begin{align*}
 \hat \beta_{\mathcal I} & = \Sigma_{[n]}^{-1} \cdot \frac{1}{n} \sum_{i=1}^{n}X_i
 \big\{ B_i X_i \hat R_i + (1-B_i) \hat R(X_i,Y_i)\big\}
 :=  \Sigma_{[n]}^{-1} \hat \phi_{[n]},
\end{align*}
where $\hat{R}(X_j,Y_j)$ are marginalized residuals based on $\hat e_{\mathcal I}(x,y)$, as defined below \eqref{equ:residuals}.

\begin{proposition}\label{prop:p_x_bar}
Under \Cref{ass:6,ass:7,ass:8}, the bias of $\hat \beta_{\mathcal I}$ satisfies
    \[
  \mathbb E [    \hat \beta_{\mathcal I}  ] - \beta  
\lesssim  \mathbb E \big[ | \hat  \mu(X) - \mu(X)| \big] + \mathbb E \big[ | \hat  e_{\mathcal I  }(X,Y) - e(X,Y)|  \big] + O(1/\sqrt{n}).
    \] 
\end{proposition}
By the definition of $\hat e_{\mathcal I}$ and the Lipschitz property of the sigmoid function, we have
\[
\mathbb E [ | \hat  e_{\mathcal I }(X,Y) - e(X,Y)|]\lesssim \mathbb E [| \hat  \mu(X) - \mu(X)|] +\mathbb E [| \hat  p (X,Y) - p(X,Y)|] + 
O(1/\sqrt{n}),
\]
where we omit higher-order terms involving products of errors from $\hat \mu$ and $\hat e_{\mathcal I}$. Thus, the bias of $\hat \beta_{\mathcal I}$ is of the same order as that of $\hat \beta_{\mathcal I,\hat p}$ defined above. Turning to variance,
\[
\Var[ \hat \beta_{\mathcal I}] = n^{-2} \mathbb E \Big\{\Sigma_{[n]}^{-1} X_{[n]}^{\top}\Var\big(  \hat V_{[n]}  \mid  X_{[n]}\big)  X_{[n]} \Sigma_{[n]}^{-1}\Big\},
\]
where $\hat V_{[n]} = (V_1, \dots, V_n)^{\top}$, and $V_i := B_i \hat R_i + (1 - B_i) \hat R(X_i, Y_i)$. Without using inverse probability weights, $V_i$ likely has lower variance than $\hat p^{-1}(X_i,Y_i)\hat  R_i$ used in  $\hat \beta_{\mathcal I,\hat p}$.

\section{Technical proofs}\label{sect:technical}

\subsection{Proof of \texorpdfstring{\Cref{prop:ej}}{Proposition~\ref{prop:ej}}}

\begin{proof}
  In the setup of \Cref{example:1}, using Bayes' rule,
\begin{align*}
e(x,y) & = \mathbb P\{Z_i = 1 \mid X_i=x,Y_i=y\} \\
& =  \mathbb P\{Z_i = 1, X_i=x,Y_i=y\}/ \mathbb P\{X_i=x,Y_i=y\} \\ 
& = \frac{f_{Y_i\mid X_i,Z_i}(y\mid x,1)e(x)}{f_{Y_i\mid X_i,Z_i}(y\mid x,1)e(x) +f_{Y_i\mid X_i,Z_i}(y\mid x,0)[1-e(x)]} \\ 
& = \left(1+ \frac{f_{Y_i\mid X_i,Z_i}(y\mid x,0)}{f_{Y_i\mid X_i,Z_i}(y\mid x,1)}\right)^{-1},
\end{align*}
where $f_{Y_i\mid X_i,Z_i}$ is the density function of $Y_i$ given $X_i$ and $Z_i$. The treatment assignment probability $e(X_i)=1/2$ for any value of $X_i$ in the Bernoulli design in \Cref{example:1}.

We next introduce \citet{robinson1988root}'s transformation of the linear model in \eqref{equ:normal_model}. Taking an expectation of both sides of the model conditional on $X_i$,
\[
\mu(X_i)  = \mathbb E[Y_i\mid X_i] = \mu_0(X_i) + e(X_i)\tau(X_i).
\]
Subtracting this from the original model yields
\[
Y_i = \mu(X_i) + [Z_i - e(X_i)]\tau(X_i) + \epsilon_i.
\]
By the normal assumption of $\epsilon_i$ in \eqref{equ:normal_model} and $e(X_i)=1/2$, the outcome $Y_i$ conditional on $X_i=x$ and $Z_i=z$ follows a normal distribution:
\[
\mathcal{N}\left(\phi(x,z):=\mu(x)+ [z -e(x)]\tau(x),\nu^2\right),
\]
where the variance $\nu^2$ does not depend on the value $z$ of $Z_i.$
This allows us to simplify the expression of $e(x,y)$ above by rewriting the density ratio,
\begin{align*}
\frac{f_{Y_i\mid X_i,Z_i}(y\mid x,0)}{f_{Y_i\mid X_i,Z_i}(y\mid x,1)}   & =  \exp\left\{  \frac{\phi(x,1) - \phi(x,0)}{\nu^2} y + \frac{h^2(x,0)-h^2(x,1)}{2\nu^2}\right\} \\
& = \exp\left\{  \frac{-\tau(x)}{\nu^2} y + \frac{\tau(x)\mu(x)}{\nu^2}\right\} \\
& = \exp\left\{-[y-\mu(x)]\tau(x)/\nu^{2}\right\}.
\end{align*}
This gives the expression of $e(X_j,Y_j)$ in \eqref{equ:ej} using $\sigma(t) = 1/\{1+\exp(-t)\}$.

\end{proof}

\subsection{Proof of \texorpdfstring{\Cref{thm:expected_power}}{Proposition~\ref{thm:expected_power}}}\label{sect:proof_thm_1}

As mentioned before \Cref{thm:expected_power}, the $p$-value in \eqref{equ:p_k} can be approximated by a bivariate normal integral involving two Gaussian variables, $\tilde T_k$ and $T_k$:
\begin{equation}\label{equ:ttk}
\begin{split}
&T_k \sim \mathcal{N}\Big(E_k := \sum_{j\in \mathcal J_k} \hat  W_j e(X_j,Y_j),\ V_k:= \sum_{j\in \mathcal J_k}\hat  W_j^2 e(X_j,Y_j)[1-e(X_j,Y_j)]  \Big), \\
& \tilde T_k \sim \mathcal{N}\Big(\tilde E_k := \sum_{j\in \mathcal J_k} \hat  W_j e(X_j),\ \tilde V_k:= \sum_{j\in \mathcal J_k}\hat  W_j^2 e(X_j)[1-e(X_j)]  \Big).
\end{split}
\end{equation}
where $e(X_i)$ is the treatment assignment probability for unit $i$, as defined below \eqref{equ:two_mu}.

Based on the means and variances defined in \eqref{equ:ttk}, \Cref{thm:expected_power_2} provides a Gaussian approximation of the $p$-value in \eqref{equ:p_k}; its proof is given in \Cref{sect:thm_2_proof}. In \Cref{sect:expression}, we then verify the expressions of the means and variances in \Cref{thm:expected_power}.

\vspace{1pt}
\begin{theorem}\label{thm:expected_power_2}
 Under \Cref{assumption:weight,assumption:delta}, the  $p$-value in \eqref{equ:p_k} satisfies that
\begin{equation}\label{equ:p_k_gaussian_2}
\hat P_k (X_{\mathcal{J}_k},Y_{\mathcal{J}_k}) =  1 -  \Phi\bigg( \hat f_k(\mathcal J_k ): = \left[ V_{k}  + \tilde V_{k}\right]^{-1/2}\left[ E_{k} - \tilde E_k \right]  \bigg)  +  O_{\mathbb P }\left(|\mathcal J_k|^{-1/2}\right).
\end{equation}
where $\Phi(\cdot)$ is the cumulative distribution function of the standard normal.
\end{theorem}

\subsubsection{Proof of \texorpdfstring{\Cref{thm:expected_power_2}}{Proposition~\ref{thm:expected_power_2}}}
\label{sect:thm_2_proof}

\begin{proof}

The $p$-value $\hat P_k (O_{\mathcal J_k})$ in \eqref{equ:pk} can be expressed as a function of the difference
\[
T(\tilde O_{\mathcal J_k}) - T( O_{\mathcal J_k}) =  \sum_{j\in \mathcal{J}_k} \hat W_j \tilde Z_j  -  \sum_{j\in \mathcal{J}_k} \hat W_j  Z_j,
\]
This difference can be rewritten, without changing its sign, as
\begin{align*}
    &\ \tilde V_{k}^{-1/2} \sum_{j\in \mathcal{J}_k} \hat W_j  \big[ \tilde Z_j -e(X_j) \big]   - \tilde V_{k}^{-1/2} \sum_{j\in \mathcal{J}_k} \hat W_j  \big[  Z_j -e(X_j) \big] \\
=  &\ \tilde V_{k}^{-1/2} \tilde A_k
-\tilde V_{k}^{-1/2} A_k +  \tilde V_{k}^{-1/2}\sum_{j\in \mathcal{J}_k} \hat W_j \big[ e(X_j) - e(X_j,Y_j) \big],
\end{align*}
where 
$\tilde A_k = \sum_{j\in \mathcal{J}_k} \hat W_j  \big[ \tilde Z_j -e(X_j) \big]$ and  
$A_k = \sum_{j\in \mathcal{J}_k} \hat W_j  \big[ Z_j -e(X_j,Y_j) \big]$.
Then,
\begin{equation}\label{equ:p_k_proof}
\hat P_k(O_{\mathcal{J}_k})
    =  \tilde{\mathbb P}\bigg\{ \tilde V_{k}^{-1/2} \tilde A_k \geq \tilde V_{k}^{-1/2} A_k +  \tilde V_{k}^{-1/2}\sum_{j\in \mathcal{J}_k} \hat W_j \big[  e(X_j,Y_j) - e(X_j) \big] \bigg\}.
\end{equation}
Conditional on $X_{[n]},Y_{[n]}$ and $Z_{[n]},$ the remaining randomness in the $p$-value comes from the randomized treatment  assignments $\tilde Z_{\mathcal{J}_k}$ in the summation  $\tilde A_k$ defined above.
It is straightforward to verify that $\hat  W_j[\tilde Z_j -e(X_j)],j\in \mathcal J_k,$ are independent and mean-zero random variables with finite variance and third absolute moment. Also, 
\[
\tilde V_k  =   \sum_{j\in \mathcal J_k}\hat  W_j^2 e(X_j)[1-e(X_j)] >0.
\]
By the Berry-Essen theorem in Chapter 3 of \citet{chen2010normal}, we have 
\begin{align*}
\sup_{t\in \mathcal R} \left|    \mathbb P_{\tilde Z_{\mathcal{J}_k}}\left\{ \tilde V_{k}^{-1/2} \tilde A_k \leq t \right\}  - \Phi(t)      \right| & \leq \frac{C}{\tilde V_k^{3/2} }\sum_{j\in \mathcal J_k}\hat W_j^3\mathbb E_{\tilde Z_{j}}\left\{ |\tilde Z_j - e(X_j)|^3 \right\}  \\
& = O_{\mathbb P }\bigg(1/\sqrt{|\mathcal J_k|}\bigg).
\end{align*}
where  $C$ is a universal constant. The equality holds under \Cref{assumption:weight,assumption:delta}.
By $E_k - \tilde E_k =\sum_{j\in \mathcal{J}_k} \hat W_j \big[  e(X_j,Y_j) - e(X_j) \big]$, we can rewrite \eqref{equ:p_k_proof} as
\[
\hat P_k(O_{\mathcal{J}_k}) = 1 - \Phi\left(  \tilde V_{k}^{-1/2} V_{k}^{1/2} V_{k}^{-1/2} A_k +   \tilde V_{k}^{-1/2} \left[E_k - \tilde E_k\right]   \right)  + \ O_{\mathbb P }\bigg(1/\sqrt{|\mathcal J_k|}\bigg).
\]
Conditional on $X_{[n]},Y_{[n]}$ and $Z_{\mathcal I}$, the remaining randomness in the $p$-value comes from the observed treatment  assignments $Z_{\mathcal J_k}$ in the summation $A_k$ defined above.
Similar to the proof above, we can first check that $\hat  W_j[Z_j -e(X_j,Y_j)],j\in \mathcal J_k,$ are independent and mean-zero random variables with finite variance and third absolute moment. Also, it is easy show that 
$\Phi( \tilde V_{k}^{-1/2} V_{k}^{1/2} x + E_k - \tilde E_k)$ is Lipschitz function of $x$ because its derivative with respect to $x$ is bounded by $\tilde V_{k}^{-1/2} V_{k}^{1/2}/\sqrt{2\pi}.$ 

By the Berry-Essen theorem in \citet[Theorem 3.1]{chen2010normal}, we have
\[
\hat P_k(X_{\mathcal{J}_k}, Y_{\mathcal{J}_k})  = 1 -  \mathbb E_{Z\sim \mathcal{N}(0,1)} \left\{    \Phi\left(  \tilde V_{k}^{-1/2} V_{k}^{1/2} Z +   \tilde V_{k}^{-1/2}  \left[E_k - \tilde E_k\right]   \right)\right\} + O_{\mathbb P }\Big(1/\sqrt{|\mathcal J_k|}\Big).
\]
By the well-known bivariate normal integral identity,
\begin{align*}
\hat P_k(X_{\mathcal{J}_k}, Y_{\mathcal{J}_k}) &  = 1 -  \Phi\left(   \left[1+   \tilde V_{k}^{-1} V_{k}  \right]^{-1/2} \tilde V_{k}^{-1/2}  \left[E_k - \tilde E_k\right]  \right)  +  O_{\mathbb P }\Big(1/\sqrt{|\mathcal J_k|}\Big) \\
&  =  1 -  \Phi\bigg( \left[ V_{k}  + \tilde V_{k}\right]^{-1/2}\left[ E_{k} - \tilde E_k \right]  \bigg)  +  O_{\mathbb P }\Big(1/\sqrt{|\mathcal J_k|}\Big).
\end{align*}

\end{proof}

\subsubsection{Expressions of the means and variances in \texorpdfstring{\Cref{thm:expected_power}}{Proposition~\ref{thm:expected_power}}}
\label{sect:expression}
\begin{proof}
Under the assumptions we make, the weight $\hat W_j$ can be written as \eqref{equ:w_j_hat} as
\[
\hat W_j = 2\left[2 Y_j - \mu_1(X_j) - \mu_0(X_j)  \right] = 4[Y_j - \mu(X_j)],
\]
using the expression of $\mu_1$ and $\mu_0$ in \eqref{equ:two_mu}. Observe that 
\[
e(X_j,Y_j) = \sigma([Y_j-\mu(X_j)]\tau(X_j)/\nu^2)\geq 1/2,
\]
if $Y_j - \mu(X_j)$ and $\tau(X_j)$ have the same sign, and less than $1/2$ otherwise. Then,
\begin{align*}
&\ \tilde W_j \big[e(X_j,Y_j) - e(X_j)\big] =   4 [Y_j - \mu (X_j)] \cdot [\sigma([Y_j-\mu(X_j)]\tau(X_j)/\mu^2)-1/2] \\ \vspace{5pt}
= &  
\begin{cases}
    4 |Y_j - \mu (X_j)| |e(X_j,Y_j)-1/2|,   \hspace{9pt} \text{ if } \text{sign}(Y_j-\mu(X_j))= \text{sign}(\tau(X_j))=1,\\ \vspace{5pt}
        4 |Y_j - \mu (X_j)| |e(X_j,Y_j)-1/2|,  \hspace{9pt}  \text{ if } \text{sign}(Y_j-\mu(X_j))=-1, \text{sign}(\tau(X_j))= 1. \\ \vspace{5pt}
        -4 |Y_j - \mu (X_j)| |e(X_j,Y_j)-1/2|, \text{ if } \text{sign}(Y_j-\mu(X_j))=1, \text{sign}(\tau(X_j))= -1. \\  \vspace{5pt}     
              -4 |Y_j - \mu (X_j)||e(X_j,Y_j)-1/2|,    \text{ if } \text{sign}(Y_j-\mu(X_j))= \text{sign}(\tau(X_j))=-1,\\
\end{cases} \\
& = \text{sign}(\tau(X_j))\cdot |Y_j-\mu(X_j)|\cdot |e(X_j,Y_j) - 1/2|.
\end{align*}
Substituting this into the definitions of $E_k$ and $\tilde E_k$ in \eqref{equ:ttk},
we obtain the expression of $E_j-E_j$ in the proposition.
Using the expression of $\hat W_j$ above, we can rewrite the variances sum $V_k+\tilde V_k$ defined in \eqref{equ:ttk} as follows:
\begin{align*}
V_k + \tilde  V_k &= 16 \sum_{j\in \mathcal{J}_k}[Y_j-\mu(X_j)]^2 \cdot\left\{e(X_j,Y_j)[1-e(X_j,Y_j)]  +1/4\right\} \\
& = 4 \sum_{j\in \mathcal{J}_k}[Y_j-\mu(X_j)]^2 \cdot\left\{2-C_j^2\right\},
\end{align*}
as required in the proposition. 
\end{proof}

\subsection{Proof of Proposition \ref{prop:threshold}}

\begin{proof}
    Let $\weight^*$ denote an optimizer of $\hat{l}_k(\weight)$.  
    We first observe that if $\tau(X_j) \le 0$, then $\weight^*_j = 0$.  
    Indeed, if $\weight^*_j > 0$ for some $j$ with $\tau(X_j) \le 0$, setting $\weight^*_j$ to zero would strictly increase the objective $\hat{l}_k(\weight)$, contradicting the optimality of $\weight^*$.

    Now consider the case where $\tau(X_j) > 0$. Given any optimizer $\weight^*$ of $\hat{l}_k(\xi)$, define $B^* :=  V_k(\weight^*) + \tilde{V}_k(\weight^*)$.
It is straightforward to verify that  $\weight^*$ also solves the following constrained optimization problem:
\begin{equation}\label{proof:eq:alternative.optimization}   
\begin{split}
\max_{\weight \in [0,1]^{|\mathcal S_k|}} \quad & E_k(\weight) - \tilde{E}_k(\weight), \\
\text{subject to} \quad & V_k(\weight) + \tilde{V}_k(\weight) \le B^*.
\end{split}
\end{equation}
For if $\weight^*$ were not an optimizer of \eqref{proof:eq:alternative.optimization}, then there would exist a feasible point $\weight'$ satisfying the constraint in \eqref{proof:eq:alternative.optimization} and
$E_k(\weight') - \tilde{E}_k(\weight') > E_k(\weight^*) - \tilde{E}_k(\weight^*).$
This implies  $\hat{l}_k(\weight') > \hat{l}_k(\weight^*)$,
contradicting the optimality of $\weight^*$ for $\hat{l}_k(\xi)$.

In the main text, the quantities $E_k(\weight)$, $V_k(\weight)$, $\tilde{E}_k(\weight)$, and $\tilde{V}_k(\weight)$ are defined as
    \begin{align*}
        E_k(\weight) :=& \sum_{j \in \mathcal S_k} 
        \weight_j \hat{W}_j e(X_j,Y_j),\quad V_k(\weight):= 
        \sum_{j\in \mathcal S_k} \weight_j \hat{W}_j^2 e(X_j,Y_j)[1-e(X_j,Y_j)], \\
        \tilde{E}_k(\weight) :=& \sum_{j \in \mathcal S_k} \weight_j \hat  W_j e(X_j),\quad
        \tilde{V}_k(\weight):= \sum_{j\in \mathcal S_k} \weight_j \hat  W_j^2 e(X_j)[1-e(X_j)].
    \end{align*}
Using these expressions, the optimization problem in \eqref{proof:eq:alternative.optimization} can be rewritten as
\begin{align*}
    \max_{\weight \in [0,1]^{|\mathcal{S}_k|}} \quad & \sum_{j \in \mathcal{S}_k} \weight_j \cdot \left( \operatorname{sign}(\tau(X_j))  \cdot |Y_j - \mu(X_j)|  \cdot|e(X_j, Y_j) - 1/2| \right), \\
    \text{subject to} \quad & \sum_{j \in \mathcal{S}_k} \weight_j \cdot \left( [Y_j - \mu(X_j)]^2  \cdot \left\{ 1/2 - [e(X_j, Y_j) - 1/2]^2 \right\} \right) \le B^*/4.
\end{align*}
By \Cref{lemm:fractional.knapsack}, this problem admits an optimal solution that takes the threshold form described in \Cref{prop:threshold}.
\end{proof}

\begin{lemma}[\citet{neyman1933ix}]\label{lemm:fractional.knapsack}
Let $a, b \in \mathbb{R}^m$ with $b_i > 0$ for all $i$, and let $B \in (0, \sum_{i=1}^m b_i]$.  
Then the optimization problem
\[
\max_{w \in [0,1]^m} \sum_{i=1}^m w_i a_i, \quad \text{subject to} \quad \sum_{i=1}^m w_i b_i \le B
\]
admits an optimal solution of the form
\[
w_i^* =
\begin{cases}
1, & \text{if } a_i / b_i > c \text{ and } a_i \ge 0, \\
c', & \text{if } a_i / b_i = c \text{ and } a_i \ge 0, \\
0, & \text{if } a_i / b_i < c \text{ or } a_i < 0,
\end{cases}
\]
for some $c \ge 0$ and  $c' \in [0,1)$.
\end{lemma}

\begin{proof}[Proof of \Cref{lemm:fractional.knapsack}]
    As in the proof of \Cref{prop:threshold}, any index $i$ with $a_i < 0$ must have $w_i^* = 0$ in any optimal solution, since including it would reduce the objective.

For indices with $a_i \ge 0$, define $\rho_i := a_i / b_i$, and without loss of generality, assume they are sorted in decreasing order: $\rho_1 \ge \rho_2 \ge \cdots \ge \rho_m$. Let $w^*$ be defined as follows: set $w_i^* = 1$ for the largest $\rho_i$ values until the constraint holds with equality. Let $k$ be the smallest index such that
$\sum_{i=1}^k b_i \ge B$. Let
$c = \rho_k$  and $ c' = [B - \sum_{i=1}^{k-1} b_i]/b_k \in [0,1).$
Then set $w_k^* := c'$ and $w_i^* = 0$ for all $i > k$.
For any feasible point $w \in [0,1]^m$, 
\[
        w_i^* a_i- \rho_k w_i^* b_i - w_i a_i + \rho_k w_i b_i =  (w_i^* - w_i)(a_i- \rho_k b_i) \ge 0.
\]
    Because  $w_i^* = 1 \ge w_i$ for $a_i > \rho_k b_i, w_i^* \le w_i$ for $a_i < \rho_k b_i$, and \( c = \rho_k \) implies $a_k - \rho_k b_k = 0.$ Finally, summing up the equations for $i\in [m]$ proves the claim, given that the constraint $\sum_{i=1}^{m}w_ib_i \leq B$ should hold with equality for any maximizer $w$.
\end{proof}

\subsubsection{Derivative calculation}\label{sect:derivative}

In \Cref{prop:threshold}, the function $h(X_j,Y_j)$ can be written as $h_1(C_j)/h_0(X_j,Y_j)$, where
\[
h_0(X_j,Y_j) :=  |Y_j-\mu(X_j)| \quad \text{and} \quad h_1(C_j) :=   C_j/(2-C_j^2).
\]
Consider that $|Y_j - \mu(X_j)| \approx |\tau(X_j)|/2$ when ignoring the error $\epsilon_j$ in \eqref{equ:normal_model}. By \eqref{equ:ej},
\begin{equation}\label{equ:approx_h}
h_0(X_j,Y_j) = \left| \tau^{-1}(X_j)\nu^2  \text{logit}\{e(X_j, Y_j)\} \right| \approx \left| \text{logit}\{e(X_j, Y_j)\} \nu^2/2\right|^{1/2}.
\end{equation}
The probability $e(X_j,Y_j)$ can be expressed as
\begin{align*}
 e(X_j,Y_j)   =  \begin{cases}
        [1-C_j]/2, & \text{if } e(X_j,Y_j)\in (0,1/2), \\
        [1+C_j]/2, & \text{if } e(X_j,Y_j)\in (1/2,1).
    \end{cases}
\end{align*}
In either case, we can re-express $h_0(X_j,Y_j)$ using
\[
\left| \text{logit}\{e(X_j, Y_j)\} \right| = \ln \left\{ \frac{1+C_j}{1-C_j} \right\}:=h_2(C_j).
\]
Then, we have the following derivative expression:
\[
\frac{d h(X_j,Y_j) }{dC_j} \propto \frac{d}{dC_j}\left[\frac{h_1(C_j)}{h_2(C_j)}\right] = \frac{[2-C_j^2]h_2(C_j) - C_j [ -2C_j h_2(C_j) + (2-C_j^2)   h_2'(C_j)/2  ]}{[2-C_j^2]^2 h_2^{3/2}(C_j) }.
\]
The denominator is positive. The numerator can be written as 
\[
[2+C_j^2] h_2(C_j) - C_j[2-C_j^2]/[1-C_j^2],
\]
which is positive unless $C_j \approx 1$. Nevertheless, units with $C_j \approx 1$ are likely to be included in the inference fold regardless.

\subsection{Proof of \texorpdfstring{\Cref{prop:r_learner}}{Proposition~\ref{prop:r_learner}}}

\begin{proof}
By the definitions of the residuals $\hat R_i$ and  $R(X_j,Y_j,z)$ above the proposition, we can write the objective function in \eqref{equ:br_learner} as 
\begin{align*}
&  \sum_{i\in \mathcal I } \left\{Y_i - \hat \mu (X_i) - [Z_i-e(X_i)]\tau(X_i) \right\}^2 + \sum_{j\in \mathcal J }\sum_{z=0}^{1}[1-z + (2z-1)e(X_j,Y_j) ] \\
& \hspace{232pt} \cdot \left\{Y_j - \hat \mu (X_j) - [z-e(X_j)]\tau(X_j) \right\}^2 \\
= &  \sum_{i\in \mathcal I } w_i \left\{\hat R_i - \tau(X_i) \right\}^2 + \sum_{j\in \mathcal J }\sum_{z=0}^{1}[1-z + (2z-1)e(X_j,Y_j) ] \\
& \hspace{232pt} \cdot w_j(z) \left\{\hat R(X_j,Y_j,z) -\tau(X_j) \right\}^2.
\end{align*}
Because $e(X_i)=1/2$ for all $i\in [n]$ in the Bernoulli design, all the weights,
\[
 w_i = (Z_i-e(X_i))^2 , \ w_i(0) = (0-e(X_i))^2 \  \text{ and }  \  w_i(1) = (1-e(X_i))^2,
\]
are equal to 1/2. Dropping these constant weights from the objective, we have
\[
\mathcal L = \sum_{i\in \mathcal I }\left\{\hat R_i - \tau(X_i) \right\}^2 + \sum_{j\in \mathcal J }\sum_{z=0}^{1}\big[1-z + (2z-1)e(X_j,Y_j) \big] \left\{\hat R(X_j,Y_j,z) -\tau(X_j) \right\}^2.
\]
Observe that $\sum_{z=0}^{1}[1-z+(2z-1)e(X_j,Y_j)] = 1-e(X_j,Y_j) +e(X_j,Y_j)=1$.
For a linear model $\tau(x) = x^{\top }\beta$, 
the loss $\mathcal L$ is minimized by a least-squares fit,
\begin{align*}
 \hat \beta_{\mathcal I }  & = \left(  \sum_{i\in \mathcal I}X_iX_i^{\top}  +  \sum_{j\in \mathcal J }X_jX_j^{\top} \sum_{z=0}^{1}\big[1-z + (2z-1)e(X_j,Y_j) \big]              \right)^{-1} \\
 &\hspace{20pt} \cdot \bigg\{ 
    \sum_{i\in \mathcal I }X_i \hat R_i + \sum_{j\in \mathcal J }X_j\sum_{z=0}^{1}\big[1-z + (2z-1)e(X_j,Y_j) \big] \hat R(X_j,Y_j,z)\bigg\}. \\
    & = (X_{[n]}^\top X_{[n]})^{-1}\bigg\{ 
    \sum_{i\in \mathcal I }X_i \hat R_i + \sum_{i\in \mathcal J  }X_j 
    \hat R(X_j,Y_j)\bigg\},
\end{align*}
by the definition of $\hat R(X_j,Y_j)$ above the proposition.

We next prove the expression of the difference between $\hat \beta_{\mathcal I}$ and $\hat \beta_{[n]}$.

By \eqref{eq:solution},  $\hat \beta_{[n]} = (X_{[n]}^\top X_{[n]})^{-1}\sum_{i=1}^{n}X_{i}\hat R_i$. Define $\hat D_j = Y_j - \hat \mu(X_j)$. Then,
\begin{align*}
&\ \mathbb E \big\{ \|\hat \beta_{[n]}- \hat \beta_{\mathcal I }\|^2   \mid X_{[n]}, Y_{[n]}\big\} \\
= & \sum_{j\in[n]\setminus \mathcal I } X_j^{\top} (X_{[n]}^\top X_{[n]})^{-2}X_{j}\cdot \mathbb E \big\{ [\hat R_j -    \hat R(X_j,Y_j)]^2 \mid X_j,Y_j \big\} \\
= & \sum_{j\in[n]\setminus \mathcal I} X_j^{\top} (X_{[n]}^\top X_{[n]})^{-1}X_{j} \cdot \big(\mathbb E \big\{ \hat R_j^2\mid X_j,Y_j \big\} -    \hat R^2(X_j,Y_j) \big) \\
= & \sum_{j\in[n]\setminus \mathcal I} X_j^{\top} (X_{[n]}^\top X_{[n]})^{-1}X_{j} \cdot \big(4 \hat D_j^2 -   [2e(X_j,Y_j)\hat D_j - 2 (1-e(X_j,Y_j))\hat D_j ]^2 \big) \\
= & \sum_{j\in[n]\setminus \mathcal I} X_j^{\top} (X_{[n]}^\top X_{[n]})^{-1}X_{j} \cdot 4 \hat D_j^2 \big( 1 -   [2 e(X_j,Y_j) - 1]^2 \big) \\
= &\ 4\sum_{j\in[n]\setminus \mathcal I} X_j^{\top} (X_{[n]}^\top X_{[n]})^{-1}X_{j} \cdot  \hat [Y_j - \hat \mu (X_j)]^2 \big( 1 -   C_j^2 \big),
\end{align*}
as required.
\end{proof}

\subsection{Proof of \texorpdfstring{\Cref{prop:variance}}{Proposition~\ref{prop:variance}}}
Recall that if the success probability of a Bernoulli random variable is given by a sigmoid function $\sigma(t)$, then its variance $\sigma(t)(1 - \sigma(t))$ is also the derivative of the sigmoid function. Using this fact and the sigmoid expression in \Cref{prop:ej}, 
\[
e(X_j,Y_j)[1-e(X_j,Y_j)] = \sum_{t=1}^{\infty}  (-1)^{t+1} t \exp \left\{-t \tau(X_j)[Y_j-\mu(X_j)] /\nu^2 \right\},
\]
when $\tau(X_j) [Y_j-\mu(X_j)]\geq0$.
The second equality is obtained by Taylor expansion of the sigmoid function.
Since the variance is a symmetric function of $\tau(X_j) [Y_j-\mu(X_j)]$, the same expansion holds if $\tau(X_j) [Y_j-\mu(X_j)]<0$, which completes the proof.

\subsection{Proof of \texorpdfstring{\Cref{thm:validity}}{Proposition~\ref{thm:validity}}}

\begin{proof}
    
We follow the partition-based idea in \citet{zhang2023randomization} to provide the validity of our randomization tests. Here, we partition the space $\mathcal Z$ of assignments $ Z_{[n]}$ into disjoint subsets based on the selected nuisance fold $\mathcal I = \mathcal I (X_{[n]},Y_{[n]},Z_{[n]})$:
\[
\mathcal{A}_I(X_{[n]},Y_{[n]}) =  \{ z_{[n]}\in \mathcal Z: \mathcal I (X_{[n]},Y_{[n]},z_{[n]})  = I \},
\]
for all values $I$ of $\mathcal I$.
Every subset $\mathcal{A}_I$ satisfies an invariance property: if $z_{I} =  z_{I}'$,
\begin{equation}\label{equ:invariance_1}
z_{[n]} \in \mathcal{A}_I(X_{[n]},Y_{[n]}) \  \Longrightarrow \  z_{[n]}' \in \mathcal{A}_I(X_{[n]},Y_{[n]}),
\end{equation}
Equation \eqref{equ:invariance_1} means that if the selected nuisance fold \(\mathcal{I}(X_{[n]}, Y_{[n]}, Z_{[n]}) = I\) under \(Z_{[n]} = z_{[n]}\), it remains \(I\) under \(Z_{[n]} = z_{[n]}'\) for any \(z_{[n]}'\) such that \(z_{I}' = z_{I}\). 

Next, we note that $ Z_{[n]}$ is randomized independently of the potential outcomes:
\begin{equation}\label{equ:ci_random}
Y_{[n]}(0),Y_{[n]}(1)\independent Z_{[n]}\mid X_{[n]}.
\end{equation}
Under the nulls $H_{0,k}, k\in \mathcal K_0$, and by \cref{ass:consistency},
\begin{equation}\label{equ:invariance_2}
Y_i  = Y_i(1) = Y_i(0),\forall i \in \mathcal S_{\mathcal K_0}.
\end{equation}
Denote \( \mathcal{S}_{\mathcal{K}_z} = \bigcup_{k \in \mathcal{K}_z} \mathcal{S}_k \) for $z\in \{0,1\}$. The previous two equations imply that 
\begin{align*}
&\ Y_{\mathcal S_{\mathcal K_0}},Y_{\mathcal S_{\mathcal K_1}}(0),Y_{\mathcal S_{\mathcal K_1}}(1) \independent Z_{\mathcal S_{\mathcal K_0}},Z_{\mathcal S_{\mathcal K_1}}\mid X_{[n]} \\
\Rightarrow &\ Y_{\mathcal S_{\mathcal K_0}},Y_{\mathcal S_{\mathcal K_1}} \independent Z_{\mathcal S_{\mathcal K_0}\setminus I }\mid X_{[n]},Z_{I \cup \mathcal S_{\mathcal K_1} } \\
\Rightarrow &\ Y_{\mathcal S_{\mathcal K_0}\setminus I} \independent Z_{\mathcal S_{\mathcal K_0}\setminus I}\mid X_{[n]},Y_{I \cup \mathcal S_{\mathcal K_1}},Z_{I \cup \mathcal S_{\mathcal K_1}}.
\end{align*} 
Then by the invariance described in \eqref{equ:invariance_1} and \eqref{equ:invariance_2}, we have
\[ Y_{\mathcal S_{\mathcal K_0}\setminus I}\independent Z_{\mathcal S_{\mathcal K_0}\setminus I }\mid X_{[n]}, Y_{I \cup \mathcal S_{\mathcal K_1}},Z_{I \cup \mathcal S_{\mathcal K_1}}, Z_{[n]}\in  \mathcal{A}_{I}(X_{[n]},Y_{[n]}).
\]
This means the assignments of the null units in the inference fold no longer depend on these units' observed outcomes once we fix the nuisance fold. Then, 
\begin{align*}
 &\ Z_{\mathcal S_{\mathcal K_0}\setminus I } \mid X_{[n]}, Y_{[n]}, Z_{I \cup\mathcal S_{\mathcal K_1}}, Z_{[n]}\in  \mathcal{A}_{I}(X_{[n]},Y_{[n]})  \\
   \overset{d}{=} &\ Z_{\mathcal S_{\mathcal K_0}\setminus I }\mid X_{[n]},Y_{I \cup \mathcal S_{\mathcal K_1}}, Z_{I \cup \mathcal S_{\mathcal K_1}}, Z_{[n]}\in  \mathcal{A}_{I}(X_{[n]},Y_{[n]}).  
\end{align*} 
Since this holds for any value of $\mathcal I = \mathcal I(X_{[n]},Y_{[n]},Z_{[n]})$, and $\mathcal S_{\mathcal K_z} \setminus  \mathcal I = \mathcal J_{\mathcal K_z}$ for $z\in \{0,1\}$,
\begin{equation}\label{equ:penultimate}
    \begin{split}
 &\ Z_{\mathcal J_{\mathcal K_0}} \mid X_{[n]}, Y_{[n]}, Z_{\mathcal I \cup \mathcal J_{\mathcal K_1}}, \mathcal I (O_{[n]})  \\
   \overset{d}{=} &\ Z_{\mathcal J_{\mathcal K_0} }\mid X_{[n]},Y_{\mathcal I \cup \mathcal J_{\mathcal K_1}}, Z_{\mathcal I \cup \mathcal J_{\mathcal K_1}}, \mathcal I (O_{[n]}) \\
   \overset{d}{=} &\ Z_{\mathcal J_{\mathcal K_0} }\mid X_{\mathcal J_{\mathcal K_0} }, \mathcal I (O_{[n]}),
   \end{split}
\end{equation}
by the Bernoulli assign, where treatment assignments are independent across units.
Using the standard validity proof of randomization tests, for example, Theorem 1 in \citet{zhang2023randomization}, we can show that
\begin{align*}
     \mathbb P \{ \hat P_k(O_{\mathcal{J}_k}) \leq \alpha \mid X_{[n]},Y_{[n]}, Z_{\mathcal I \cup \mathcal J_{\mathcal K_1} }, \mathcal I(O_{[n]}) \}   \leq  \alpha.
\end{align*}
Randomization $p$-values computed via Monte-Carlo simulation remain valid because the randomized and observed assignments are exchangeable. For more details, see Theorem 2 in both \citet{hemerik2018exact} and \citet{ramdas2023permutation}.

Conditioning on the selected nuisance fold \( \mathcal{I}(O_{[n]}) \), the assignments \( Z_{\mathcal{J}_{\mathcal{K}_0}} \) in the null inference fold are independent Bernoulli random variables. Moreover, conditioning on \( X_{[n]}, Y_{\mathcal{I} \cup \mathcal{J}_{\mathcal{K}_1}} \), and \( Z_{\mathcal{I} \cup \mathcal{J}_{\mathcal{K}_1}} \) fixes all other sources of randomness, including the estimators of \( \mu \) and \( \tau \), involved in the null \( p \)-values \( \hat{P}_k(O_{\mathcal{J}_k}) \). This establishes the joint independence among the null $p$-values stated in the theorem.
\end{proof}

\subsection{Proof of \texorpdfstring{\Cref{prop:p_x}}{Proposition~\ref{prop:p_x}}}

\begin{proof}
Let $\Sigma = \mathbb E[X_i X_{i}^{\top}]$.
We decompose the error of $\hat \beta_{\mathcal I,\hat p}$ as follows:
\begin{equation}\label{equ:three_terms}
    \begin{split}
    \hat \beta_{\mathcal I,\hat p}-\beta  & = (\Sigma_{\mathcal I,\hat p}^{-1}- \Sigma_{\mathcal I,p}^{-1}+ \Sigma_{\mathcal I,p}^{-1}-\Sigma^{-1}+ \Sigma^{-1} )\phi_{\mathcal I,\hat p}  - \Sigma^{-1}\Sigma\beta \\
    & = (\Sigma_{\mathcal I,\hat p}^{-1}- \Sigma_{\mathcal I,p}^{-1})\phi_{\mathcal I,\hat p} 
    + (\Sigma_{\mathcal I,p}^{-1} - \Sigma^{-1})\phi_{\mathcal I,\hat p}
    + \Sigma^{-1}(\phi_{\mathcal I,\hat p} -\Sigma\beta).
 \end{split}
\end{equation}
In the first term in the last line,
\begin{align*}
\|\Sigma_{\mathcal I,\hat p}- \Sigma_{\mathcal I,p}\|_{\text{op}} &\ = \left\|\frac{1}{n}\sum_{i=1}^{n}B_i X_i X_i^{\top}\left[\hat p^{-1}(X_i,Y_i) - p^{-1}(X_i,Y_i) \right]\right\|_{\text{op}}    \\
&\ \leq n^{-1}\sum_{i=1}^{n}B_i \|X_i\|^2 | \hat p^{-1}(X_i,Y_i) - p^{-1}(X_i,Y_i) |   \\
&\ \lesssim n^{-1}
\sum_{i=1}^{n} | \hat p(X_i,Y_i) - p(X_i,Y_i) |.
\end{align*}
Then,
\begin{align*}
\mathbb E\left[  (\Sigma_{\mathcal I,\hat p}^{-1}- \Sigma_{\mathcal I,p}^{-1})\phi_{\mathcal I,\hat p}   \right]
= &\
\mathbb E\left[  \Sigma_{\mathcal I, p}^{-1}(\Sigma_{\mathcal I, p}-\Sigma_{\mathcal I,\hat p})\Sigma_{\mathcal I,\hat p}^{-1} \phi_{\mathcal I,\hat p}   \right]   \\
\leq &\ \mathbb E\left[  \lambda_{\min}^{-1}( \Sigma_{\mathcal I, p}) \cdot
\lambda_{\min}^{-1}( \Sigma_{\mathcal I,\hat p}) \cdot \|\Sigma_{\mathcal I, p}-\Sigma_{\mathcal I,\hat p}\|_{\text{op}} \cdot \|_{\text{op}} \cdot  \|\phi_{\mathcal I,\hat p}\|   \right]     \\
\lesssim &\  \mathbb E \left[ | \hat p(X_i,Y_i) - p(X_i,Y_i) |  \right].
\end{align*}
We now upper bound the second term in \eqref{equ:three_terms}. Observe that 
\[
\Sigma_{\mathcal I,p} - \Sigma = n^{-1}\sum_{i=1}^{n} [B_i/p(X_i,Y_i)-1]X_iX_i^{\top}.
\]
Applying the matrix Bernstein inequality in \citet{vershynin2018high}[Theorem 5.4.1] to the zero-mean random matrices in the average, we have 
\[
\|  \Sigma_{\mathcal I,p} - \Sigma  \|_{\text{op}} = O_{\mathbb P}\left(\sqrt{ n^{-1}\log d}\right).
\]
Using the same proof above for the first term,
\[
\mathbb E[(\Sigma_{\mathcal I,p}^{-1} - \Sigma^{-1})\phi_{\mathcal I,\hat p} ] = 
\mathbb E\left[  \Sigma^{-1} (\Sigma -\Sigma_{\mathcal I,p})\Sigma_{\mathcal I, p}^{-1} \phi_{\mathcal I,\hat p}   \right]  = O (1/\sqrt{n}).
\]
We now turn to the last term in \eqref{equ:three_terms}. First,
\[
\mathbb E \big[\phi_{\mathcal I,\hat p}\big] -\Sigma\beta =
\mathbb E\big\{ X_i \big[ \hat p^{-1}(X_i,Y_i) B_i\hat R_i -X_i^{\top}\beta \big] \big\}.
\]
By the definition of $R_i$, we have
\begin{align*}
\hat p^{-1}(X_i,Y_i)B_i\hat R_{i}  - X_i^{\top }\beta = & [\hat p^{-1}(X_i,Y_i)  - p^{-1}(X_i,Y_i) ] B_i\hat R_{i}  + p^{-1}(X_i,Y_i) B_i(\hat  R_{i}-R_i)  \\
& + \frac{B_i-p(X_i,Y_i)}{ p(X_i,Y_i)}X_i^{\top} \beta  + \frac{B_i}{ p(X_i,Y_i)}\cdot \frac{\epsilon_i}{Z_i-e(X_i)}.
\end{align*}
Compute the expectation conditional on $X_i$:
\begin{align*}
\mathbb E\left[\frac{B_i\hat R_{i} }{\hat p(X_i)}  - X_i^{\top }\beta   \  \bigg|\ X_i \right]\lesssim &\ \mathbb E[| \hat p(X_i) - p(X_i)| \mid X_i] + \mathbb E[| \hat \mu(X_i) - \mu(X_i)| \mid X_i]\\
&\ + \mathbb E\left[ \frac{p(X_i,Y_i)-p(X_i,Y_i)}{ p(X_i,Y_i)}X_i^{\top} \beta  \  \bigg|\ X_i\right] \\
&\ +  \mathbb E\left[        \mathbb E\bigg[ \frac{B_i}{ p(X_i,Y_i)}\cdot \frac{\epsilon_i}{Z_i-e(X_i)} \  \bigg|\ X_i,Y_i,\epsilon_i,Z_i \bigg]                                     \  \bigg|\ X_i\right]  \\
= &\  \mathbb E[| \hat p(X_i) - p(X_i)| \mid X_i] + \mathbb E[| \hat \mu(X_i) - \mu(X_i)| \mid X_i]\\
&\ +  \mathbb E\left[         \frac{\mathbb E[B_i\mid X_i,Y_i]}{ p(X_i,Y_i)}\cdot \frac{\epsilon_i}{Z_i-e(X_i)}                               \  \bigg|\ X_i\right]\\
= &\  \mathbb E[| \hat p(X_i) - p(X_i)| \mid X_i] + \mathbb E[| \hat \mu(X_i) - \mu(X_i)| \mid X_i],
\end{align*}
by $(\epsilon_i,Z_i)\independent B_i \mid X_i, Y_i$
and $\mathbb E[\epsilon_i\mid X_i ] = 0$. Marginalizing out $X_i$ proves the claim. 
\end{proof}

\subsection{Proof of \texorpdfstring{\Cref{prop:p_x_bar}}{Proposition~\ref{prop:p_x_bar}}}

\begin{proof}
We first divide the error $\hat \beta_{\mathcal I} - \beta$ into two error terms.
\[
 \hat \beta_{\mathcal I} - \beta = (\Sigma_{[n]}^{-1} - \Sigma^{-1} )\hat\phi_{[n]} + \Sigma^{-1}(\hat\phi_{[n]} - \Sigma \beta ).
\]
As in the last subsection, by the matrix Bernstein inequality in \citet{vershynin2018high},
\[
\mathbb E[ (\Sigma_{[n]}^{-1} - \Sigma^{-1} )\hat\phi_{[n]} ]  
= O(1/\sqrt{n}).
\]
Denote the marginalized residual based on $\mu$ and $e(X_j,Y_j)$ as
\[
R(X_j,Y_j)  = e(X_j,Y_j) R(X_j,Y_j,1) + [1- e(X_j,Y_j)] R(X_j,Y_j,0).
\]
Next, we bound the second error term:
\begin{align*}
 \mathbb E [\hat \phi_{[n]}] - \Sigma \beta  = &\ \mathbb E\big[X_i  (B_i \hat R_i + (1-B_i) \hat R(X_i,Y_i) - X_{i}^{\top}\beta)\big] \\
= &\ \mathbb E\big[X_i  (B_i [\hat R_i-R_i] + (1-B_i) [\hat R(X_i,Y_i)-R(X_i,Y_i)]  \\
 &\ \hspace{12pt} + B_i R_i + (1-B_i)R(X_i,Y_i) - X_{i}^{\top}\beta)\big]  \\
 \lesssim &\ \mathbb E\big[|\hat \mu(X_i) - \mu(X_i)| +  |\hat e_{\mathcal I  }(X_i,Y_i) - e(X_i,Y_i)| \\
 &\ \hspace{12pt} + X_i (B_i [R_i- R(X_i,Y_i)] + R(X_i,Y_i) - X_{i}^{\top}\beta)\big]  \\
= &\ \mathbb E\big[|\hat \mu(X_i) - \mu(X_i)| +  |\hat e_{\mathcal I  }(X_i,Y_i) - e(X_i,Y_i)| \big].
\end{align*}
by $\mathbb E[R(X_i,Y_i) - X_i^{\top}\beta\mid X_i] = \mathbb E\big[\epsilon_i/[Z_i-e(X_i)]\mid X_i\big]=0$.
We also use
\begin{align*}
R_i  -  R(X_i,Y_i) =  &\  Z_i\cdot [ R(X_i,Y_i,1)  -  R(X_i,Y_i)  ]  + [1-Z_i]\cdot [ R(X_i,Y_i,0)  -  R(X_i,Y_i)  ] \\
= &\   Z_i\cdot [1-  e(X_i,Y_i)]\cdot [ R(X_i,Y_i,1)- R(X_i,Y_i,0)] \\
&\ - [1-Z_i]\cdot   e(X_i,Y_i)\cdot [ R(X_i,Y_i,1)- R(X_i,Y_i,0)] \\
= &\ [Z_i - e (X_i,Y_i) ] \cdot [ R(X_i,Y_i,1)- R(X_i,Y_i,0)]. 
\end{align*}
Taking an expectation conditional on ($X_{i},Y_{i}$) leads to
\[
\mathbb E \big[ R_i  - R(X_i,Y_i) \mid X_{i},Y_{i}\big] =   [e(X_i,Y_i) - e (X_i,Y_i) ] \cdot [R(X_i,Y_i,1)-R(X_i,Y_i,0)]  = 0.
\]
Combining the bounds for both error terms proves the claim.
\end{proof}

\section{Additional simulations}\label{sect:appendix_exp}

\subsection{Details of \texorpdfstring{\Cref{fig:AdaSplit_example}}{Proposition~\ref{fig:AdaSplit_example}}}\label{sect:appendix_exp_figure}

To generate \Cref{fig:AdaSplit_example}, we simulate a sample of \( n = 200 \) i.i.d.\ units. For each unit \( i \), we generate a covariate \( X_i \sim \mathrm{Unif}(0, 5) \), and a treatment assignment variable \( Z_i \sim \mathrm{Bernoulli}(0.5) \). The outcome \( Y_i \) is then generated as follows:
\[
Y_i = \begin{cases}
\epsilon_i, & \text{if } Z_i = 0, \\
X_i + \epsilon_i, & \text{if } Z_i = 1,
\end{cases}
\quad \text{where } \epsilon_i \sim \mathcal{N}(0, 1).
\]

\subsection{Results of BaR-learner}\label{sect:appendix_exp_consistent}

We conduct additional experiments using the default simulation setup in \Cref{sect:exp_setup} to assess the consistency of BaR-learner. For each sample size \( n \in \{300, 600, \dotsc, 1500\} \), we run the AdaSplit algorithm in \Cref{sect:alg} until 50\% of the units are assigned to the nuisance fold \( \mathcal{I} \); the remaining units form the inference fold \( \mathcal{J} \).
\Cref{fig:BaR.learner.consistency} shows that
the normalized error of the BaR-learner estimator vanishes as $n$ increases,  indicating that  the estimator converges consistently to the true CATE function \( \tau \) in \eqref{equ:tau_synthetic}.

\begin{figure}[h]
\hspace{-1.5cm}
     \centering
\includegraphics[width=0.35\textwidth]{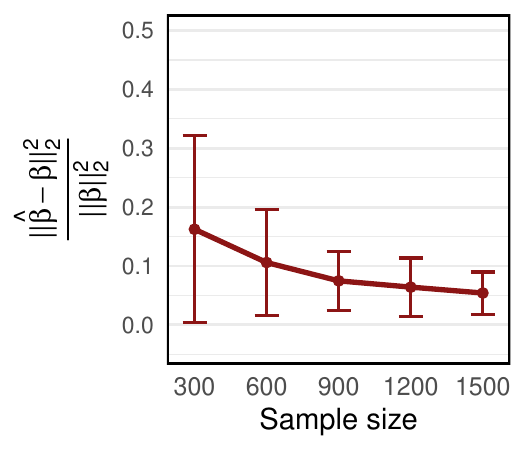}
\centering
    \caption{
Empirical consistency of BaR-learner on nuisance folds of size \( 0.5n \) selected by AdaSplit for sample size \( n \in \{300, 600, \dotsc, 1500\} \). 
    }
    \label{fig:BaR.learner.consistency}
\end{figure}

\vspace{-5pt}

\subsection{Results for XGboost}\label{sect:appendix_exp_xgboost}

In \Cref{fig:group.p.val.nuisance.estimator}, 
we assess how the choice of 
$\hat{\mu}$ impacts the test power. Compared to the default linear regression model, XGBoost provides a less accurate estimate of $\mu$ in our simulation setup in \Cref{sect:exp_setup}, where $\mu(x)$ is a linear function of $x$. Consequently, the p-values in panel (b) are slightly larger than those in panel (a).
Despite this, both panels show that \ART produces smaller p-values than \RRT. This indicates that, with a fixed $\hat{\mu}$, \ART gains power through its adaptive allocation of units, enabling them to contribute more to estimation and inference.

\begin{figure}[h]
    \hspace{2.3cm}
    \begin{minipage}{0.27\textwidth}
        \centering
        \includegraphics[clip, trim = 0cm 0cm 5cm 0cm, height = 0.9\textwidth]{plot/Comparison_setting_default.pdf}
        \subcaption*{ \hspace{5pt} (a) Linear regression.}
    \end{minipage}
    \hspace{1cm}
     \begin{minipage}{0.27\textwidth}
        \centering
        \includegraphics[clip, trim = 0cm 0cm 0cm 0cm, height = 0.9\textwidth]{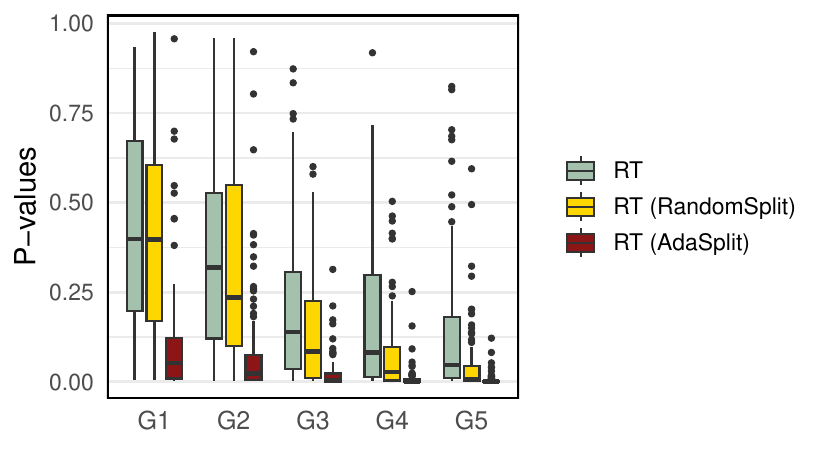}
        \subcaption*{\qquad \quad (b) XGboost.}
    \end{minipage}
    \caption{Boxplots of subgroup $p$-values. Each panel shows the results for RT, RT (RandomSplit), and RT (AdaSplit) on five subgroups (G1–G5) defined in \Cref{sect:exp_setup}. Panel (a) uses linear regression for estimating $\hat{\mu}$, while Panel (b) uses XGBoost \citep{chen2016xgboost}. The results are aggregated over $100$ trials.
    }
    \label{fig:group.p.val.nuisance.estimator}
\end{figure}

\vspace{-5pt}

\subsection{Results for varying proportions}\label{sect:appendix_exp_results}

In \Cref{fig:group.p.val.split.proportion}, we vary $\rho$, the nuisance fold proportion in random sample splitting and the maximum nuisance proportion in AdaSplit. As $\rho$ increases, \RRT becomes less powerful due to the smaller inference fold. In contrast, \ART adaptively chooses the proportion of units used for CATE estimation and may stop early when the estimator converges, preserving a sufficient number of units for inference.

\begin{figure}[h]
    \begin{minipage}{0.27\textwidth}
        \centering
        \includegraphics[clip, trim = 0cm 0cm 5cm 0cm, height = 0.9\textwidth]{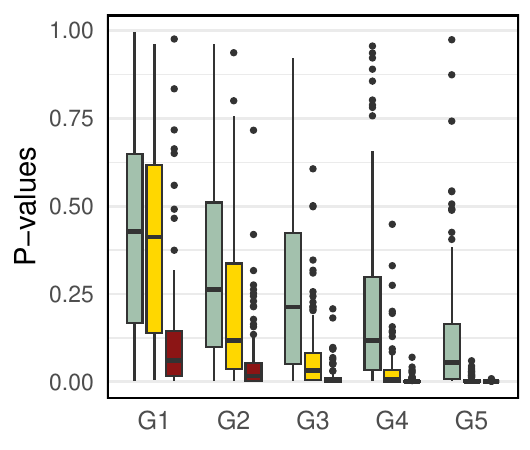}
        \subcaption*{\quad (a) $\rho = 0.25$.}
    \end{minipage}
    \hspace{0cm}
     \begin{minipage}{0.27\textwidth}
        \centering
        \includegraphics[clip, trim = 0cm 0cm 5cm 0cm, height = 0.9\textwidth]{plot/Comparison_setting_default.pdf}
        \subcaption*{\quad (b) $\rho = 0.5$.}
    \end{minipage}
    \hspace{0cm}
         \begin{minipage}{0.27\textwidth}
        \centering
        \includegraphics[clip, trim = 0cm 0cm 0cm 0cm, height = 0.9\textwidth]{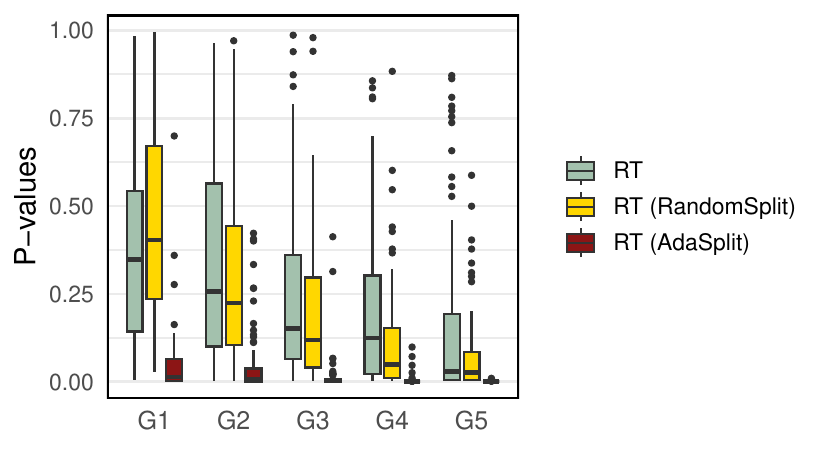}
        \subcaption*{\quad (c) $\rho = 0.75$.}
    \end{minipage}
\caption{Boxplots of subgroup $p$-values. Each panel shows the results for RT, RT (RandomSplit), and RT (AdaSplit) on five subgroups (G1–G5) defined in \Cref{sect:exp_setup}. From panels (a) to (c), both the nuisance fold proportion in RT (RandomSplit) and the maximum nuisance fold proportion in \ART increase from 0.25 to 0.75.  Each configuration is repeated 100 times, and the results are aggregated.
}
\label{fig:group.p.val.split.proportion}
\end{figure}

\end{document}